%% file: SIAM_NDEWT_main.tex
\documentclass[hidelinks]{article}
\usepackage[a4paper, total={6in, 8in}]{geometry}

\def\R{{\mathbb{R}}}
\def\N{{\mathbb{N}}}
\def\Z{{\mathbb{Z}}}

\def\Z{{\mathbb{Z}}}
\def\V{{\mathcal{V}}}
\def\F{{\mathcal{F}}}
\def\E{{\mathcal{E}}}
\def\Li{{\mathrm{L}^2}}
\def\rmd{{\mathrm{d}}}

\usepackage{comment}

\input{SIAM_NDEWT_shared}

\ifpdf
\hypersetup{
  pdftitle={Multidimensional empirical wavelet transform},
  pdfauthor={C.-G. Lucas and J. Gilles}
}
\fi

\graphicspath{{./Images}}

\begin{document}

\maketitle

% REQUIRED
\begin{abstract}
The empirical wavelet transform is a data-driven time-scale representation consisting of an adaptive filter bank. Its robustness to data has made it the subject of intense developments and an increasing number of applications in the last decade. However, it has been mostly studied theoretically for signals and its extension to images is limited to a particular generating function. This work presents a general framework for multidimensional empirical wavelet transform based on any wavelet kernel. It also provides conditions to build wavelet frames for both continuous and discrete transforms. Moreover, numerical simulations of transforms are given.
\end{abstract}

% REQUIRED

\medskip

% REQUIRED
\textbf{MSCcodes:} 42C15, 42C40, 68U10.

\section{Introduction}

The wavelet transform is a reference tool for time-scale representation used in many signal and image processing techniques, such as denoising, deconvolution and texture segmentation. Originally, it consists of projecting data onto wavelet filters that are built from a mother wavelet which is scaled and modulated independently of the data. In practice, this leads to the construction of wavelet filters based on a prescribed scheme, such as a dyadic decomposition. Although this approach is widely used in contemporary research, it is not guaranteed to be optimal for the data at hand, since a prescribed scheme does not take into account the specificity of the underlying Fourier spectrum. Therefore, data-driven filtering approaches have received much attention to provide an accurate time-scale representation that is robust to the data. 
Among them, empirical mode decomposition \cite{huang1998empirical}, a purely algorithmic approach which behaves as a filter bank \cite{flandrin2004empirical}, has been widely used in real-world applications.
Inspired by this decomposition, empirical wavelet transform has been introduced in \cite{Gilles2013} to provide a more consistant decomposition with a sound theoretical foundation \cite{kedadouche2016comparative,wardana2016comparative}.
In the last decade, this transform has been the subject of an increasing interest through a continuous development and numerous applications in various fields, as reviewed in \cite{huang2018review,liu2019recent}. To name a few applications, we can mention seismic time-frequency analysis \cite{seismic}, electrocardiogram (ECG) signal compression \cite{ecg}, epileptic seizure detection \cite{EEGEWT,panda2021epileptic}, speech recognition \cite{lavrynenko2023method}, time series forecasting \cite{mohammadi2023using}, glaucoma detection \cite{maheshwari2016automated}, hyperspectral image classification \cite{prabhakar2017two}, texture segmentation \cite{Huang2019}, multimodal medical image fusion \cite{polinati2020multimodal} and cancer histopathological image classification \cite{deo2024ensemble}. 
Particularly, this transform has shown to outperform traditional wavelet transforms for texture segmentation \cite{huang2018review} and framelet tranforms for texture denoising \cite{hurat2020empirical}.

The construction of empirical wavelet systems consists of two steps: $(i)$ extracting supports of the harmonic modes of the function under study, and $(ii)$ constructing empirical wavelet filters that are compactly (or very rapidly decaying) supported in the Fourier domain by the extracted supports. 
This construction is the core of an ongoing and active research effort for one-dimensional (1D) and two-dimensional (2D) functions.
For 1D functions, the detection of the segments supporting each harmonic modes in the Fourier domain is usually performed by extracting the lowest minima between them using a scale-space representation \cite{Gilles2014a}. For 2D functions, several different techniques have been proposed to delimit the supports of the harmonic modes, such as the Curvelet \cite{Gilles2013a}, Watershed \cite{hurat2020empirical} and Voronoi \cite{gilles2022empirical} partitioning methods.
The construction of wavelet filters based on the detected supports is usually done in the Fourier domain. In the 1D domain, \cite{cewt} proposed a general framework to build continuous empirical wavelet filters from a generating function, such as the Meyer, Shannon or Gabor scaling functions. These 1D empirical wavelet systems are written as modulations and dilations of a wavelet kernel based on segments that divide the Fourier line. Such systems have been shown to induce both continuous and discrete frames in \cite{Gilles2024frame}. In the 2D domain, empirical wavelet filters have been designed following the Littlewood-Palley wavelet formulation for various Fourier partitioning methods \cite{gilles2022empirical,Gilles2013a,hurat2020empirical}. However, the proposed construction is only valid for discrete images and is entirely based on the distance to the support boundaries, which limits its extension to classic scaling functions. So far, no construction of empirical wavelets in higher dimensions has been proposed. 
In addition, empirical wavelet filters for real-valued functions are built from supports that take into account the symmetry of the corresponding Fourier transform, but these have been little studied theoretically.

This work aims to provide a general framework for empirical wavelet transforms of multidimensional functions, thus extending the 1D framework in \cite{cewt,Gilles2024frame}. We show that we can build empirical wavelets on Fourier supports, symmetric or not, from any wavelet kernel defined in the Fourier domain, using diffeomorphisms that map these supports to the wavelet kernel's Fourier support. Both continuous and discrete transforms are considered. In addition, conditions for the construction of wavelet frames are examined.

% The outline is not required.
The paper is organized as follows. The construction of the multidimensional empirical wavelet systems and the resulting transforms is described in
\Cref{sec:ews}. Theoretical results on these systems as frames are given in \Cref{sec:ewf}. The special case of Fourier supports resulting from affine deformations of the wavelet kernel's Fourier support is studied in \Cref{sec:aff_def}. Finally, in \Cref{sec:num_exp}, specific 2D wavelet systems are tailored from classic wavelet kernels and studied numerically on images.
A Matlab toolbox will be made publicly available at the time of publication.

\section{Notations and reminders}
We respectively denote $\partial \Omega$ and $\overline{\Omega}$ the boundary and closure of a set $\Omega \in \R^N$. We denote $\Upsilon^+=\{n\in\Upsilon \mid n \geq 0\}$ the subset of positive elements of a set $\Upsilon\in \Z$. The Jacobian of a differentiable function $\gamma$ is denoted $J_{\gamma}$. We recall that a function $\gamma$ is called a diffeomorphism if it is infinitely differentiable and invertible of inverse infinitely differentiable. 

We consider that the space of square integrable functions $\Li(\R^N)$ is endowed with the usual inner product
$$ \langle f,g \rangle =\int_{\R^N} f(x) \overline{g(x)} \rmd x.$$
The Fourier transform $\widehat{f}$ of a function $f \in \Li(\R^N)$ and its inverse are given by, respectively,
\begin{align*}
& \widehat{f}(\xi) = \F (f)(\xi) = \int_{\R^N} f(x)e^{-2\pi \imath (\xi \cdot x)} \rmd x, \\
& f(x) = \F^{-1} (\widehat{f})(x) = \int_{\R^N} \widehat{f}(\xi)e^{2\pi \imath (\xi \cdot x)} \rmd \xi,
\end{align*}
where $\cdot$ stands for the usual dot product in $\R^N$. 
The translation operator $T_y$ of a function $f \in \Li(\R^N)$ for $y \in \R^N$ is defined by 
$$(T_y f)(x)= f(x-y).$$

We recall hereafter definitions on frames that are essential throughout this work. A set of functions $\{g_p\}_{p\in \V}$ of $\Li(\R^N)$ is a frame if there exist two constants $0<A_1\leq A_2<\infty$ such that, for every $f\in \Li(\R^N)$,
$$A_1 \Vert f \Vert_2 \leq \int_{ \V} \vert \langle f,g_p \rangle \vert^2 \, \rmd \mu(p) \leq A_2 \Vert f \Vert_2,$$
with $(\V,\mu)$ a measure set.
In particular, $\{g_p\}_{p\in \V}$ is called a tight frame if $A_1=A_2$ and a Parseval frame if $A_1=A_2=1$. 
A frame $\{\widetilde{g}_p\}_{p\in \V}$ is the dual frame of $\{g_p\}_{p\in \V}$ if, for every $f\in \Li(\R^N)$,
$$f(x) = \int_{\V} \langle f,g_p \rangle \, \widetilde{g}_p(x) \, \rmd \mu(p).$$
It is well known that a frame $\{g_p\}_{p\in \V}$ is a tight frame if and only if there exist $A>0$ such that $\{g_p/A\}_{p\in \V}$ is a dual frame of $\{g_p\}_{p\in \V}$.
In these definitions, the set $\V$ can be a cartesian product of both uncountable and countable sets. In particular, countable sets equipped with a counting measure lead to summations instead of integrals.
For an in-depth introduction to frames, interested readers can see \cite{christensen2008frames}.

\section{Multidimensional dimensional empirical wavelet system}
\label{sec:ews}

In this section, we build empirical wavelet systems for the analysis of a given $N$-variate function $f\in\Li(\R^N)$, with $N\in\N$. A key feature of empirical wavelet filters is that they are adaptive: they are constructed from a Fourier domain partitioning scheme that is data-driven rather than pre-specified. 
We first define this Fourier partitioning. We then provide the formalism to construct empirical wavelet systems. Finally, we define empirical wavelet transforms.

\subsection{Fourier domain partitions}

In this section, we introduce the formalism used for the Fourier supports involved in the construction of empirical wavelet systems. 

\begin{definition}
A partition of the Fourier domain is defined as a family of disjoint connected open sets $\{\Omega_n\}_{n\in\Upsilon}$, with $\Upsilon \subset \Z$, of closures covering the Fourier domain, i.e., satisfying
 $$\R^N=\bigcup_{n\in\Upsilon} \overline{\Omega_n} \quad \textrm{ and } \quad n\neq m \Rightarrow \Omega_n\cap \Omega_m=\varnothing.$$
\end{definition}

A partition can consist of either $(i)$ an infinite number of $\Omega_n$ with compact closure $\overline{\Omega_n}$ or $(ii)$ a finite number of $\Omega_n$ composed of both compact and non-compact closures $\overline{\Omega_n}$. Since the sets $\Omega_n$ are connected, the closure $\overline{\Omega_n}$ is compact if and only if $\Omega_n$ is bounded.

In the 1D domain, this definition coincides with the Fourier line partitioning proposed by \cite{cewt} where unbounded intervals, called rays, correspond to non-compact supported closures and segments to compact supported closures.
An example of a partition of the Fourier domain in the 2D domain is given in \cref{fig:partition}.

\begin{figure}[t]
\centering
  \includegraphics[width=.45\linewidth]{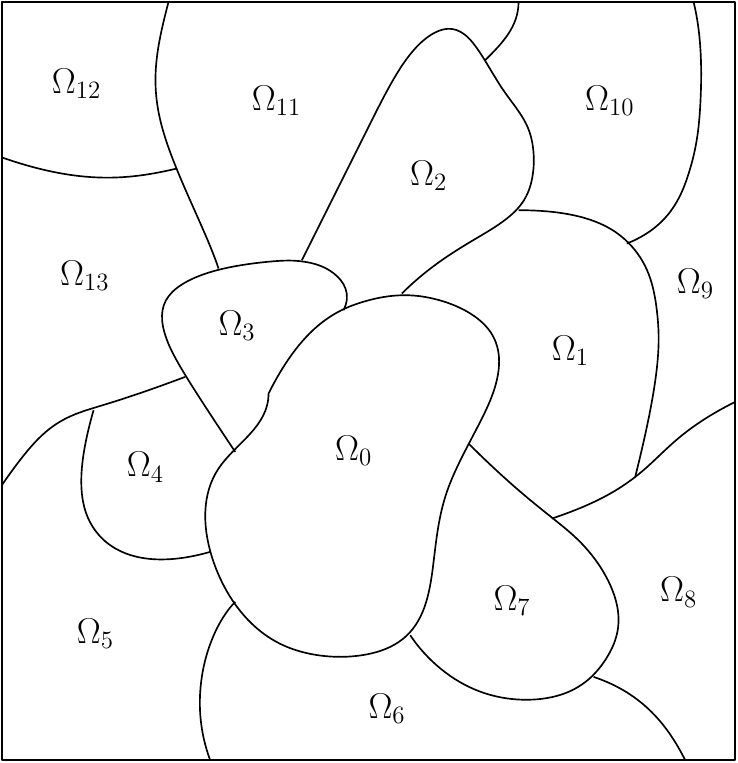}
  \caption{{\bf Partitioning.} Example of a partition of a square domain in the 2D case.}
  \label{fig:partition}
\end{figure}

The partition $\{\Omega_n\}_{n\in\Upsilon}$ is usually constructed such that each set $\Omega_n$ contains one mode of the Fourier domain.
To this end, we can first detect the modes of the Fourier spectrum using the scale-space representation \cite{Gilles2014a} and then define boundaries $\partial \Omega_n$ separating these modes.
In 1D, the intervals $\Omega_n$ can be defined using the lowest minima between these modes. In 2D, the extraction of supports $\Omega_n$ can be performed by various methods, such as the Tensor (grid) \cite{Gilles2013}, Ridgelet (radial), Curvelet (radial and angular) \cite{Gilles2013a}, Watershed \cite{hurat2020empirical} or Voronoi \cite{gilles2022empirical} tilings.

A special case is raised by the real-valued functions since their Fourier spectrum is even.
It is therefore natural to consider a symmetric partition $\{\Omega_n\}_{n\in\Upsilon}$, defined as follows.

\begin{definition}
A partition $\{\Omega_n\}_{n\in\Upsilon}$ is called symmetric if
\begin{equation*} %\label{eq:sym_partition}
n \in \Upsilon \Rightarrow -n \in \Upsilon \quad \textrm{ and } \quad \xi \in \Omega_n \Rightarrow -\xi \in \Omega_{-n}.
\end{equation*}
\end{definition}
This definition implies that the region $\Omega_0$ contains the zero frequency.
A procedure of symmetrization of partitions has been proposed in \cite{hurat2020empirical}.
For such partitions, we will build filters on sets of paired regions $\Omega_n \cup \Omega_{-n}$ rather than on single regions $\Omega_n$.

\subsection{Empirical wavelet filter bank}
In this section, we introduce empirical wavelet filter banks induced by a wavelet kernel for a given Fourier domain partition. Two types of filters are proposed, depending on the symmetry of the Fourier domain supports.

Let $\psi \in \Li(\R^N)$ be a wavelet kernel such that its Fourier transform $\widehat{\psi}$ is localized in frequency and verifies, for some connected open bounded subset $\Lambda\subseteq \mathrm{supp} \, \psi$, 
$$\exists \, 0 \leq \delta < 1, \quad \int_{\overline{\Lambda}} \left\vert \widehat{\psi}(\xi) \right\vert^2 \rmd \xi = (1-\delta) \left\lVert \widehat{\psi} \right\rVert_{\Li}.$$
This condition ensures that $\widehat{\psi}$ is mostly supported by $\overline{\Lambda}$.
Generally, $\widehat{\psi}$ is homogeneous or separable, implying that $\Lambda$ is a $1$-ball or $2$-ball in $\R^N$. 

Given a partition $\{\Omega_n\}_{n\in\Upsilon}$, the goal of this section is to define, from any wavelet kernel $\psi$, two banks of wavelet filters that are mostly supported in the Fourier domain by $(i)$ $\Omega_n$, or $(ii)$ $\Omega_n \cup \Omega_{-n}$ in the case of a symmetric partition.
To this end, we make the following assumption, which is used throughout the paper.
\begin{hypothesis}
For every $n\in\Upsilon$, there exists a diffeomorphism $\gamma_n$ on $\R^N$ such that 
$$
\begin{cases}
\displaystyle \Lambda=\gamma_n(\Omega_n)  & \textrm{if } \Omega_n \textrm{ is bounded}, \\
\displaystyle  \Lambda \subsetneq \gamma_n(\Omega_n)  & \textrm{otherwise}.
\end{cases}
$$
\end{hypothesis}

A diffeomorphism on a bounded set $\Omega_n$ is illustrated in \Cref{fig:map}. This assumption is mild since, by the Hadamard-Cacciopoli theorem \cite{de1994global}, if $\Lambda$ is simply connected (i.e., has no hole), an infinitely differentiable function $\gamma_n$ is a diffeomorphism if and only if it is proper, i.e., the preimage of any compact set under $\gamma_n$ is compact.

\begin{figure}[t]
\centering
  \includegraphics[width=.6\linewidth]{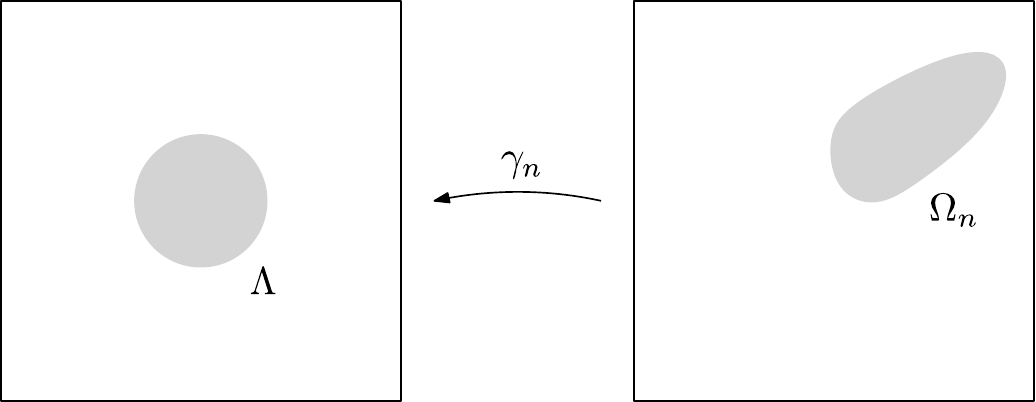}
	\caption{{\bf Mapping.} Illustration of a diffeomorphism from a set $\Omega_n$ to a disk $\Lambda$.}
  \label{fig:map}
\end{figure}

First, we consider the case of a partition $\{\Omega_n\}_{n\in\Upsilon}$ that is not necessarily symmetric.

\begin{definition}
The empirical wavelet system, denoted $\{\psi_n\}_{n\in\Upsilon}$, corresponding to the partition $\{\Omega_n\}_{n\in\Upsilon}$ is defined by, for every $\xi \in \R^N$,
\begin{equation*}%\label{eq:ews}
%\widehat{\psi}_n(\xi)= \frac{1}{\sqrt{\vert \mathrm{det} \; J_{\gamma^{-1}_n}(\gamma_n(\xi)) \vert}} \widehat{\psi}(\gamma_n(\xi)).
\widehat{\psi}_n(\xi)= \sqrt{\vert \mathrm{det}\; J_{\gamma_n}(\xi) \vert} \;  \widehat{\psi}(\gamma_n(\xi)).
\end{equation*}
\end{definition}

The determinant is a normalization coefficient for the conservation of energy when $\Lambda=\gamma_n(\Omega_n)$, i.e., 
$$\int_{\Omega_n} \left\vert \widehat{\psi}_n(\xi) \right\vert^2 \mathrm{d}\xi = \int_{\Lambda} \left\vert \widehat{\psi}(u) \right\vert^2 \mathrm{d}u.$$ 

\begin{example}\label{ex:psin_1D}
We consider the 1D case ($N=1$).
Let $\psi$ a wavelet kernel on $\R$ of Fourier transform $\widehat{\psi}$ with support $\Lambda=(-\frac{1}{2},\frac{1}{2})$ and $\{\Omega_n\}_{n\in\Upsilon}$ a family of non-overlapping open bounded intervals with center $\omega_n$ such that $\R=\bigcup_{n\in\Upsilon} \overline{\Omega_n}$. Then $\gamma_n: \xi \mapsto \frac{1}{\vert \Omega_n \vert}(\xi-\omega_n)$ are diffeomorphisms such that $\Lambda=\gamma_n(\Omega_n)$ and the empirical wavelet system is given by, for every $\xi \in \R$,
\begin{equation}
\label{eq:wavelet_bankfilters_1D}
\widehat{\psi}_n(\xi) = \frac{1}{\sqrt{\vert \Omega_n \vert}} \widehat{\psi} \left(\frac{\xi-\omega_n}{\vert \Omega_n \vert} \right).
\end{equation}
This definition is in agreement with the definition given in \cite{cewt}.
For the diffeomorphisms $\kappa_n: \xi \mapsto \frac{1}{\vert \Omega_n \vert}(\omega_n-\xi)$ also satisfying $\Lambda=\kappa_n(\Omega_n)$, Definition \ref{eq:wavelet_bankfilters_1D} becomes, for every $\xi \in \R$,
\begin{equation*}
\widehat{\psi}_n(\xi) = \frac{1}{\sqrt{\vert \Omega_n \vert}} \widehat{\psi} \left(\frac{\omega_n-\xi}{\vert \Omega_n \vert} \right),
\end{equation*}
which is different from \Cref{eq:wavelet_bankfilters_1D} if $\widehat{\psi}$ is not even.% (i.e., if $\psi$ is not a real function).
\end{example}

Now, we consider that $\{\Omega_n\}_{n\in\Upsilon}$ is a symmetric partition. To build wavelet filters $\chi_n$ which are mostly supported by $\Omega_n \cup \Omega_{-n}$, we assume that $\Lambda$ satisfies
\begin{equation*}
u \in \Lambda \Rightarrow -u \in \Lambda.
\end{equation*}
Necessarily, the system $\{\chi_n\}_{n\in\Upsilon}$ must be symmetric with respect to frequency band, that is, it must satisfy the property $$\widehat{\chi}_n = \widehat{\chi}_{-n}.$$
In this context, we only consider the diffeomorphisms $\gamma_n$ for $n \geq 0$. The function $\gamma_{-n}: \xi \mapsto -\gamma_n(-\xi)$ is a diffeomorphism, that verifies $\gamma_{-n}(\Omega_{-n}) = \Lambda $ when  $\gamma_n(\Omega_n)=\Lambda$, which suggests the following definition. 

\begin{definition}
The symmetric empirical wavelet system, denoted $\{\chi_n\}_{n\in\Upsilon}$, corresponding to the symmetric partition $\{\Omega_n\}_{n\in\Upsilon}$ is defined by, for every $\xi \in\R^N$,
\begin{equation}\label{eq:symwavelet_def}
\begin{cases}
&\displaystyle \widehat{\chi}_0(\xi) = \widehat{\psi}_0(\xi), \\
&\displaystyle \widehat{\chi}_n(\xi) = \frac{1}{\sqrt{2}} \left( \widehat{\psi}_n(\xi) + \widehat{\psi}_{-n}(\xi) \right) \quad \textrm{ for } n\neq 0,
\end{cases}
\end{equation}
where $\widehat{\psi}_n= \widehat{\psi} \circ \gamma_n$ and $\gamma_{-n}:\xi \mapsto -\gamma_n(-\xi)$.
\end{definition}

In this definition, $\widehat{\chi}_n$ is not necessarily even for any $n\in \Upsilon$, but ensures that the parity of $\widehat{\psi}_n$ implies the parity of $\widehat{\chi}_n$.
Moreover, we can write \Cref{eq:symwavelet_def} explicitely as follows, for every $n\in\Upsilon\setminus\{0\}$ and $\xi \in\R^N$,
\begin{equation*}
\widehat{\chi}_n(\xi) = \frac{1}{\sqrt{2}} \left(\sqrt{\vert \mathrm{det}\; J_{\gamma_n}(\xi) \vert} \;  \widehat{\psi}(\gamma_n(\xi)) + \sqrt{\vert \mathrm{det}\; J_{\gamma_n}(-\xi) \vert} \;  \widehat{\psi}(-\gamma_n(-\xi)) \right).
\end{equation*}
Thus, for $n\in\Upsilon\setminus\{0\}$ such that $\int_{\Omega_n \cup \Omega_{-n}} \widehat{\psi}(\gamma_n(\xi)) \overline{\widehat{\psi}(-\gamma_n(-\xi))} = 0$ and $\Lambda=\gamma_n(\Omega_n)$, the factor $1/\sqrt{2}$ guarantees the conservation of the energy, i.e., 
$$\int_{\Omega_n \cup \Omega_{-n}} \left\vert \widehat{\chi}_n(\xi) \right\vert^2 \mathrm{d}\xi = \int_{\Lambda} \left\vert \widehat{\psi}(u) \right\vert^2 \mathrm{d}u.$$ 
In particular, if $\partial \Omega_n$ and $\partial \Omega_{-n}$ are disjoint, the conservation of energy is guaranteed when $\widehat{\psi}$ has a compact support.

\begin{example} As in \Cref{ex:psin_1D}, we consider $N=1$, $\widehat{\psi}$ on $\R$ with support $\Lambda=(-\frac{1}{2},\frac{1}{2})$ and $\{\Omega_n\}_{n\in\Upsilon}$ a family of non-overlapping bounded intervals with center $\omega_n$ such that $\R=\bigcup_{n\in\Upsilon} \overline{\Omega_n}$. For the diffeomorphism $\gamma_n: \xi \mapsto \frac{1}{\vert \Omega_n \vert}(\xi-\omega_n)$, the symmetric empirical wavelet system reads, for every $\xi \in \R$,
\begin{equation*}
\begin{cases}
&\displaystyle \widehat{\chi}_0(\xi) = \frac{1}{\sqrt{ \vert \Omega_0 \vert}} \widehat{\psi} \left(\frac{\xi}{\vert \Omega_0 \vert} \right), \\
%&\displaystyle  \widehat{\chi}_n(\xi) = \frac{1}{\sqrt{2 \vert \Omega_n \vert}} \widehat{\psi} \left(\frac{\xi-\omega_n}{\vert \Omega_n \vert} \right) + \frac{1}{\sqrt{2 \vert \Omega_n \vert}} \widehat{\psi} \left(\frac{\xi+\omega_n}{\vert \Omega_n \vert} \right) \quad \textrm{ for } n \neq 0.
&\displaystyle  \widehat{\chi}_n(\xi) = \frac{1}{\sqrt{2 \vert \Omega_n \vert}} \left[ \widehat{\psi} \left(\frac{\xi-\omega_n}{\vert \Omega_n \vert} \right) + \widehat{\psi} \left(\frac{\xi+\omega_n}{\vert \Omega_n \vert} \right) \right] \quad \textrm{ for } n \neq 0.
\end{cases}
\end{equation*}
\end{example}

\begin{remark}
Due to the linearity of the inverse Fourier tranform,  the symmetric empirical wavelet system $\{\chi_n\}_{n\in\Upsilon}$ also satisfies, in the spatial domain, for every $x \in \R^N$,
\begin{equation}\label{eq:spatial_domain_symwavelet}
\chi_0(x) = \psi_0(x) \quad \textrm{ and } \quad \chi_n(x) = \frac{1}{\sqrt{2}} \left( \psi_n(x) + \psi_{-n}(x) \right) \quad \textrm{ for } n \neq 0.
\end{equation}
It is therefore sufficient to write $\psi_n$ in the spatial domain to write $\chi_n$ in the spatial domain.
\end{remark}

\subsection{Empirical wavelet transform}

In this section, we introduce continuous and discrete transforms of a function $f\in \Li(\R^N)$ based on either the empirical wavelet systems $\{\psi_n\}_{n\in\Upsilon}$ or the symmetric empirical wavelet systems $\{\chi_n\}_{n\in\Upsilon}$.

\begin{definition}\label{def:continuous_transf}
The $N$-dimensional continuous empirical wavelet transform of a  real or complex-valued function $f\in \Li(\R^N)$ is defined by, for every $n\in\Upsilon$ and $b\in\R^N$,
\begin{equation}\label{eq:cont_wave_transf}
\E_\psi^f(b,n)=\langle f, T_b\psi_n\rangle.
\end{equation}

The $N$-dimensional continuous symmetric empirical wavelet transform of a real-valued function $f\in \Li(\R^N)$ is defined by, for every $n\in\Upsilon$ and $b\in\R^N$,
\begin{equation}\label{eq:cont_wave_symtransf}
\E_\chi^f(b,n)=\langle f, T_b\chi_n\rangle.
\end{equation}
\end{definition}

In this definition, the $N$-dimensional continuous symmetric empirical wavelet transform $\E_\chi^f$ is symmetric for the frequency variable $n$ but not for the translation variable $b$.

The following proposition shows that the $N$-dimensional continuous empirical wavelet transform can be rewritten as a filtering process. It is a straight generalization of Proposition~1 of \cite{cewt}. We will adopt the notation $\psi^*(t)\equiv \psi(-t)$, and denote $\star$ and $\bullet$ the convolution and pointwise product of functions, respectively.

\begin{proposition}[Filtering process]\label{prop-ndewt}
The $N$-dimensional continuous empirical wavelet transform $\E_{\psi}^f(b,n)$ is equivalent to the convolution of $f$ with the function $\overline{\psi^*_n}$, i.e., for every $n\in\Upsilon$ and $b\in\R^N$,
\begin{equation}
\label{eq:convolution_transform}
\E_{\psi}^f(b,n) = \left(f\star\overline{\psi^*_n}\right)(b)= \F^{-1}\left(\widehat{f}\bullet\overline{\widehat{\psi}_n}\right)(b).
\end{equation}

In addition, if $f$ is a real-valued function, for every $n\in\Upsilon$ and $b\in\R^N$,
\begin{equation}
\label{eq:convolution_symtransform}
\E_{\chi}^f(b,n) = \left(f\star\overline{\chi^*_n}\right)(b)= \F^{-1}\left(\widehat{f}\bullet\overline{\widehat{\chi}_n}\right)(b).
\end{equation}
\end{proposition}
\begin{proof}
Given functions $f$ and $\psi$, we have,
\begin{align*}
  \E_{\psi}^f(b,n)= \left\langle f,T_b\psi_n \right\rangle = \int_{\R^N} f(x)\overline{T_b\psi_n(x)}\mathrm{d}x &=\int_{\R^N} f(x)\overline{\psi_n(x-b)}\mathrm{d}x\\
&=\int_{\R^N} f(x)\overline{\psi^*_n(b-x)}\mathrm{d}x\\
&=\left(f\star\overline{\psi^*_n}\right)(b).
\end{align*}
This proves the first equality of \Cref{eq:convolution_transform}. Now, noticing that
\begin{align*}
\widehat{\overline{\psi^*_n}}(\xi)= \int_{\R^N} \overline{\psi_n(-x)}e^{-2\pi \imath (\xi \cdot x)}\mathrm{d}x
& =\overline{\int_{\R^N} \psi_n(-x)e^{2\pi \imath (\xi \cdot x)}\mathrm{d}x} \\
& =\overline{\int_{\R^N} \psi_n(x)e^{-2\pi \imath (\xi \cdot x)}\mathrm{d}x}=\overline{\widehat{\psi}_n}(\xi), 
\end{align*}
we can rewrite the convolution obtained above as a pointwise multiplication in the Fourier domain,
\begin{align*}
\E_{\psi}^f(b,n)=\F^{-1}\left(\F\left(f\star\overline{\psi^*_n}\right)\right)(b)= \F^{-1}\left(\widehat{f}\bullet\widehat{\overline{\psi^*_n}}\right)(b)= 
\F^{-1}\left(\widehat{f}\bullet\overline{\widehat{\psi}_n}\right)(b).
\end{align*}
This provides the second equality of \Cref{eq:convolution_transform}.

Given the relation between $\chi_n$ and $\psi_n$ in the spatial domain, which is given by \Cref{eq:spatial_domain_symwavelet}, \Cref{eq:convolution_symtransform} directly stems from \Cref{eq:convolution_transform} using the linearity of the convolution, the inner product and the inverse Fourier transform.
\end{proof}

\begin{definition}\label{def:discrete_transf}
Let $\{b_n\}_{n\in\Upsilon} \in \R\setminus\{0\}$. 
The $N$-dimensional discrete empirical wavelet transform of a real or complex-valued function $f\in \Li(\R^N)$ is defined by, for every $n\in\Upsilon$ and $k \in \mathbb{Z}^N$,
\begin{equation*}
\E_\psi^f(b_n k,n)=\langle f, T_{b_n k}\psi_n\rangle.
\end{equation*}

The $N$-dimensional discrete symmetric empirical wavelet transform of a real-valued function $f\in \Li(\R^N)$ is defined by, for every $n\in\Upsilon$ and $k \in \mathbb{Z}^N$,
\begin{equation*}
\E_\chi^f(b_n k,n)=\langle f, T_{b_n k}\chi_n\rangle,
\end{equation*}
with $b_n=b_{-n}$.
\end{definition}

\section{Frames of empirical wavelets}
\label{sec:ewf}

In this section, we provide conditions to build empirical wavelet frames for both continuous and discrete empirical wavelet transforms of a given function $f \in \Li(\R^N)$. In particular, we examine the potential reconstruction of $f$.

\subsection{Continuous frames}

In this section, we build dual frames for the two proposed systems
$$\{T_{b}\psi_n\}_{(n,b)\in\Upsilon \times \R^N} \quad \textrm{ and } \quad \{T_{b}\chi_n\}_{(n,b)\in\Upsilon^+ \times \R^N},$$
with $\Upsilon^+ = \{n \in \Upsilon \mid n \geq 0\}$, involved in the continuous wavelet transform of \Cref{def:continuous_transf}. This allows to find sufficients conditions for these systems to be tight frames.

First, we consider the system $\{\psi_n\}_{n\in\Upsilon}$.
The following theorem guarantees the exact reconstruction of a function $f$ from the continuous empirical wavelet transform $\E_{\psi}^f$ given by \Cref{eq:cont_wave_transf}. It is a straight generalization of Proposition~2 of \cite{cewt} to the $N$-dimensional case.

\begin{theorem}[Continuous dual frame]\label{theo:icewt}
Let assume that, for $a.e.\, \xi \in \R^N$,
$$\displaystyle 0<\sum_{n\in\Upsilon}\left|\widehat{\psi}_n(\xi)\right|^2<\infty.$$
Then, for every $x \in \R^N$,
\begin{align}\label{eq:reconstruction}
f(x) = \sum_{n\in\Upsilon}\left(\E_{\psi}^f(\cdot,n)\star\widetilde\psi_n\right)(x) = \sum_{n\in\Upsilon}\int_{b \in \R^N}\langle f, T_b\psi_n \rangle T_b\widetilde{\psi}_n(x) \mathrm{d}b,
\end{align}
where the set of dual empirical wavelets $\{\widetilde\psi_n\}_{n\in\Upsilon}$ is defined by, for every $n\in\Upsilon$ and $\xi\in\R^N$,
\begin{equation*}
\widehat{\widetilde\psi_n}(\xi)=\frac{\widehat{\psi}_n(\xi)}{\displaystyle\sum_{n\in\Upsilon}\left|\widehat{\psi}_n(\xi)\right|^2}.
\end{equation*}
\end{theorem}
\begin{proof}
Using the Fourier transform and its inverse, we can write
\begin{align*}
\sum_{n\in\Upsilon}\left(\E_{\psi}^f(\cdot,n)\star\widetilde\psi_n\right)&=\F^{-1}\left(\F\left(\sum_{n\in\Upsilon}\left(\E_{\psi}^f(\cdot,n)\star\widetilde\psi_n\right)\right)\right)\\
%&=\F^{-1}\left(\sum_{n\in\Upsilon}\F\left(\E_{\psi}^f(\cdot,n)\star\widetilde\psi_n\right)(\xi)\right)\\
&=\F^{-1}\left(\sum_{n\in\Upsilon}\widehat{\E_{\psi}^f}(\cdot,n) \bullet \widehat{\widetilde\psi_n}\right)\\
&=\F^{-1}\left(\sum_{n\in\Upsilon}\widehat{f} \bullet \overline{\widehat{\psi}_n} \bullet \widehat{\widetilde\psi_n}\right)\\
&=\F^{-1}\left(\widehat{f} \bullet \sum_{n\in\Upsilon}\overline{\widehat{\psi}_n} \bullet \frac{\widehat{\psi_n}}{\sum_{n\in\Upsilon}\left|\widehat{\psi}_n\right|^2}\right)\\
%&=\F^{-1}\left(\hat{f}(\xi)\frac{\sum_{n\in\Upsilon}\left|\widehat{\psi}_n(\xi)\right|^2}{\sum_{n\in\Upsilon}\left|\widehat{\psi}_n(\xi)\right|^2}\right)\\
&=\F^{-1}\left(\widehat{f}\right)=f.
\end{align*}
This proves the first equality of \Cref{eq:reconstruction}.
Moreover, we can rewrite 
\begin{align*}
\sum_{n\in\Upsilon}\left(\E_{\psi}^f(\cdot,n)\star\widetilde\psi_n\right)(x) &= \sum_{n\in\Upsilon}\int_{b \in \R^N}\E_{\psi}^f(b,n)\widetilde{\psi}_n(x-b) \mathrm{d}b \\
&= \sum_{n\in\Upsilon}\int_{b \in \R^N}\langle f, T_b\psi_n \rangle T_b\widetilde{\psi}_n(x) \mathrm{d}b.
\end{align*}
This proves the second equality of \Cref{eq:reconstruction}.
\end{proof}

A particular case of the previous theorem is given by the following corollary.

\begin{corollary}[Continuous tight frame]\label{coro:tightewt}
If, for $a.e.\, \xi \in \R^N$,
$$\displaystyle 0<\sum_{n\in\Upsilon}\left|\widehat{\psi}_n(\xi)\right|^2 =A<\infty,$$
then $\{T_b\psi_n\}_{(n,b)\in\Upsilon \times \R^N}$ is a continuous tight frame. Specifically, for every $x \in \R^N$,
\begin{equation*}
f(x)= \frac{1}{A}\sum_{n\in\Upsilon}\int_{b \in \R^N}\langle f, T_b\psi_n \rangle T_b\psi_n(x) \mathrm{d}b.
\end{equation*}
\end{corollary}
\begin{proof}
From \Cref{theo:icewt} with 
\begin{equation*}
 \widehat{\widetilde\psi_n}(\xi)=\frac{\widehat{\psi}_n(\xi)}{\sum_{n\in\Upsilon}\left|\widehat{\psi}_n(\xi)\right|^2}=\frac{\widehat{\psi}_n(\xi)}{A}, 
\end{equation*}
it follows that
\begin{equation*}
f(x) = \frac{1}{A}\sum_{n\in\Upsilon}\left(\E_{\psi}^f(\cdot,n)\star\psi_n\right)(x) = \frac{1}{A}\sum_{n\in\Upsilon}\int_{b \in \R^N}\langle f, T_b\psi_n \rangle T_b\psi_n(x) \mathrm{d}b.
\end{equation*}
\end{proof}

Now, we consider the symmetric wavelet filter bank $\{\chi_n\}_{n\in\Upsilon^+}$. The following theorem guarantees the exact reconstruction of a real-valued function $f$ from the continuous symmetric empirical wavelet transform $\E_{\chi}^f$ given in \Cref{eq:cont_wave_symtransf}.

\begin{theorem}[Continuous symmetric dual frame]\label{theo:symicewt}
Let assume that, for $a.e.\, \xi \in \R^N$,
$$\displaystyle 0<\sum_{n\in\Upsilon^+}\left|\widehat{\psi}_n(\xi)\right|^2<\infty \quad \textrm{ and } \quad \sum_{n\in\Upsilon^+\setminus\{0\}}\left|\widehat{\psi}_n(\xi)\right| \left|\widehat{\psi}_{-n}(\xi)\right|<\infty.$$
Then, for any $x \in \R^N$,
\begin{align}
\label{eq:symreconstruction}
f(x) = \sum_{n\in\Upsilon^+}\left(\E_{\chi}^f(\cdot,n)\star \tilde \chi_n\right)(x)  = \sum_{n\in\Upsilon^+}\int_{b \in \R^N}\langle f, T_b\chi_n \rangle T_b\widetilde{\chi}_n(x) \mathrm{d}b,
\end{align}
where the set of dual symmetric empirical wavelets $\{\widetilde\chi_n\}_{n\in\Upsilon}$ is defined by, for every $n\in\Upsilon$ and $\xi\in\R^N$,
\begin{equation*}
\widehat{\widetilde\chi_n}(\xi)=\frac{\widehat{\chi}_n(\xi)}{\displaystyle\sum_{n\in\Upsilon^+}\left|\widehat{\chi}_n(\xi)\right|^2}.
\end{equation*}
\end{theorem}
\begin{proof}
First, we have
\begin{align*}
\sum_{n\in\Upsilon^+}\left|\widehat{\chi}_n(\xi)\right|^2 &= \left|\widehat{\psi}_0(\xi)\right|^2 + \frac{1}{2} \sum_{n\in\Upsilon^+\setminus\{0\} }\left|\widehat{\psi}_n(\xi)+\widehat{\psi}_{-n}(\xi)\right|^2 \\
&= \frac{1}{2} \sum_{n\in\Upsilon^+}\left|\widehat{\psi}_n(\xi)\right|^2+\frac{1}{2}\sum_{n\in\Upsilon^+}\left|\widehat{\psi}_{-n}(\xi)\right|^2+\sum_{n\in\Upsilon^+\setminus\{0\}}\left|\widehat{\psi}_n(\xi)\right| \left|\widehat{\psi}_{-n}(\xi)\right| < \infty,
\end{align*}
which ensures that the symmetric empirical wavelet filters $\widetilde\chi_n$ are well defined. Then, \Cref{eq:symreconstruction} stems from the same computation as in the proof of \Cref{theo:icewt}.
\end{proof}

Similarly to \Cref{coro:tightewt}, the following corollary gives a particular case of the previous theorem.

\begin{corollary}[Continuous symmetric tight frame]\label{coro:tightsymewt}
If, for $a.e.\, \xi \in \R^N$,
$$\displaystyle 0<\sum_{n\in\Upsilon^+}\left|\widehat{\psi}_n(\xi)\right|^2=A<\infty \quad \textrm{ and } \quad \sum_{n\in\Upsilon^+\setminus\{0\}}\left|\widehat{\psi}_n(\xi)\right| \left|\widehat{\psi}_{-n}(\xi)\right|=B<\infty,$$
then $\{T_b\chi_n\}_{(n,b)\in\Upsilon^+ \times \R^N}$ is a continuous tight frame. Specifically, for every $x \in \R^N$,
\begin{equation*}
f(x)= \frac{1}{A+B}\sum_{n\in\Upsilon^+}\int_{b \in \R^N}\langle f, T_b\chi_n \rangle T_b\chi_n(x) \mathrm{d}b.
\end{equation*} 
\end{corollary}

\begin{proof}
We can write
\begin{align*}
\sum_{n\in\Upsilon^+}\left|\widehat{\chi}_n(\xi)\right|^2 &= \frac{1}{2}\sum_{n\in\Upsilon^+}\left|\widehat{\psi}_n(\xi)\right|^2+\frac{1}{2}\sum_{n\in\Upsilon^+}\left|\widehat{\psi}_{-n}(\xi)\right|^2+\sum_{n\in\Upsilon^+\setminus\{0\}}\left|\widehat{\psi}_n(\xi)\right| \left|\widehat{\psi}_{-n}(\xi)\right| = A+B,
\end{align*}
and the result follows from the fact that $\widetilde{\chi}_n = \chi_n/(A+B)$.
\end{proof}

\subsection{Discrete frames}

In this section, we exhibit conditions to build discrete wavelet frames involved in the wavelet transforms of \Cref{def:discrete_transf}. These conditions depend on the compactness of the support of the wavelet kernel's Fourier transform $\widehat{\psi}$. The two underlying cases are examined separately.

\subsubsection{Compactly supported \texorpdfstring{$\boldsymbol{\widehat{\psi}}$}{Lg}}

In this section, we consider that $\widehat{\psi}$ has a compact support. 
Excluding the supports of a partition $\{\Omega_n\}_{n\in\Upsilon}$ with non-compact closure, we can state sufficient and necessary conditions for which the systems 
$$\{T_{b_nk}\psi_n\}_{(n,k)\in\Upsilon_{\rm comp} \times \Z^N} \quad \textrm{ and } \quad \{T_{b_nk}\chi_n\}_{(n,k)\in\Upsilon_{\rm comp}^+ \times \Z^N},$$
are Parseval frames, with
$$\Upsilon_{\rm comp} = \{n\in \Upsilon \mid  \overline{\Omega_n} \textrm{ is compact}\} \; \textrm{ and } \; \Upsilon_{\rm comp}^+ = \{n\in \Upsilon^+ \mid  \overline{\Omega_n} \textrm{ is compact}\}. $$

The following theorem first gives a sufficient and necessary condition to build a tight frame from $\{\psi_n\}_{n\in\Upsilon_{\rm comp}}$. It is a straight generalization of Theorems~4-7 of \cite{Gilles2024frame} to the $N$-dimensional case.

\begin{theorem}[Discrete Parseval frame]
Let us denote $\mathrm{L}_{\rm comp}^2(\R^N) = \{f \in \mathrm{L}^2(\R^N) \mid \mathrm{supp} \, \widehat{f} \subseteq \Gamma_{\rm comp} \}$ and $\Gamma_{\rm comp} = \underset{n \in \Upsilon_{\rm comp}}{\bigcup} \overline{\Omega_n}$.
The system $\{T_{b_nk}\psi_n\}_{(n,k)\in\Upsilon_{\rm comp} \times \Z^N}$ is a Parseval frame for $\mathrm{L}_{\rm comp}^2(\R^N)$ if and only if, for $ a.e.\, \xi \in \R^N $,
\begin{equation*}
\sum_{n \in \Upsilon_\alpha} \frac{1}{\vert b_n \vert} \widehat{\psi}_n(\xi) \overline{\widehat{\psi}_n(\xi+\alpha)} = \delta_{\alpha,0},
\end{equation*}
for every $\alpha \in \mathcal{K}$, where
$$\mathcal{K} = \underset{n \in \Upsilon_{\rm comp}}{\bigcup}b_n^{-1} \mathbb{Z}^N, \qquad \Upsilon_\alpha = \{n\in\Upsilon_{\rm comp} \mid b_n \alpha  \in\mathbb{Z}^N\}, $$
and $\delta_{\alpha,0}$ stands for the Kronecker delta function on $\R^N$, i.e., $\delta_{\alpha,0} = 1$ if $\alpha=0$ and $\delta_{\alpha,0} = 0$ otherwise.
\end{theorem}

\begin{proof}
Let us denote $\mathcal{D} = \{f \in \mathrm{L}^2(\R^N) \mid \widehat{f}\in \mathrm{L}^\infty(\R^N) \textrm{ and } \mathrm{supp}\, \widehat{f} \subset \Gamma_{\rm comp}\}.$
Theorem 2.1 in \cite{hernandez2002unified} states, with $g_p=\psi_n$, $C_p=b_n$ and $\mathcal{P}= \Upsilon_{\rm comp} $, the desired equivalence under the condition that
$$\forall f \in \mathcal{D}, \;  \sum_{n \in \Upsilon_{\rm comp}} \sum_{k\in \mathbb{Z}^N} \int_{\mathrm{supp}\, \widehat{f}} \left\vert \widehat{f}(\xi + b_n^{-1}k) \right \vert^2 \frac{1}{\vert b_n \vert} \left\vert \widehat{\psi}_n(\xi) \right\vert^2 \mathrm{d}\xi < \infty.$$
The condition above is given by the proof of Theorem~4 in \cite{Gilles2024frame}, replacing $\mathbb{Z}$ by $\mathbb{Z}^N$. This results comes from the fact that $\mathrm{supp} \, \widehat{f}$ and $\mathrm{supp} \, \widehat{\psi}_n$ are compact and that there are finitely many $\mathrm{supp}\, \widehat{\psi}_n$ that intersect $\mathrm{supp}\, \widehat{f}$. This proves the equivalence.
\end{proof}

\begin{remark}
The previous theorem implicitely permits to easily build dual frames $\widetilde{\psi}_n$ when, for $ a.e. \, \xi \in \R^N$,
$$\sum_{n\in\Upsilon} \frac{1}{\vert b_n \vert} \left|\widehat{\psi}_n(\xi)\right|^2<\infty,$$
as follows: %, for every $\xi \in \R^N$,
\begin{equation*}
\widehat{\widetilde\psi_n}(\xi)=\frac{\widehat{\psi}_n(\xi)}{\displaystyle\sum_{n\in\Upsilon} \frac{1}{\vert b_n \vert} \left|\widehat{\psi}_n(\xi)\right|^2}.
\end{equation*}
\end{remark}

Similarly, the following theorem gives a sufficient and necessary condition to build a tight frame from the symmetric filter bank $\{\chi_n\}_{n\in\Upsilon_{\rm comp}^+}$.

\begin{theorem}[Discrete symmetric Parseval frame]
Let us denote $\mathrm{L}_{\rm comp}^2(\R^N) = \{f \in \mathrm{L}^2(\R^N) \mid \mathrm{supp} \, \widehat{f} \subseteq \Gamma_{\rm comp} \}$ and $\Gamma_{\rm comp} = \underset{n \in \Upsilon_{\rm comp}}{\bigcup} \overline{\Omega_n}$.
Then, the system $\{T_{b_nk}\chi_n \}_{(n,b)\in\Upsilon_{\rm comp}^+ \times \Z^N}$ is a Parseval frame for $\mathrm{L}_{\rm comp}^2(\R^N)$ if and only if, for $ a.e.\, \xi \in \R^N$,
\begin{equation}
\sum_{n \in \Upsilon_\alpha^+} \frac{1}{\vert b_n \vert} \widehat{\chi}_n(\xi) \overline{\widehat{\chi}_n(\xi+\alpha)} = \delta_{\alpha,0} ,
\end{equation}
for every $\alpha \in \mathcal{K}^+$, where
$$\mathcal{K}^+ = \underset{n \in \Upsilon_{\rm comp}^+}{\bigcup}b_n^{-1} \mathbb{Z}^N, \qquad \Upsilon_\alpha^+ = \{n\in\Upsilon_{\rm comp}^+ \mid b_n \alpha  \in\mathbb{Z}^N\}, $$
and $\delta_{\alpha,0}$ stands for the Kronecker delta function on $\R^N$.

\end{theorem}

\begin{proof}
Let us denote $\mathcal{D} = \{f \in \mathrm{L}^2(\R^N) \mid \widehat{f} \in \mathrm{L}^\infty(\R^N) \textrm{ and } \mathrm{supp}\, \widehat{f} \subset \Gamma_{\rm comp}\}.$
Noticing that $$\Gamma_{\rm comp} = \underset{n \in \Upsilon_{\rm comp}^+}{\bigcup} \left( \overline{\Omega_n}\cup \overline{\Omega_{-n}} \right).$$
Theorem 2.1 in \cite{hernandez2002unified} states (by taking $g_p=\chi_n$, $C_p=b_n$ and $\mathcal{P}= \Upsilon_{\rm comp}^+$) that the system $\{T_{b_nk}\chi_n \}_{(n,k)\in\Upsilon^+ \times \Z^N}$ is a Parseval frame for $\mathrm{L}_{\rm comp}^2(\R^N)$ if and only if, for $ a.e.\, \xi \in \R^N$,
\begin{equation*}
\sum_{n \in \Upsilon_\alpha^+} \frac{1}{\vert b_n \vert} \widehat{\chi}_n(\xi) \overline{\widehat{\chi}_n(\xi+\alpha)} = \delta_{\alpha,0},
\end{equation*} 
under the condition that
$$\forall f \in \mathcal{D}, \;  \sum_{n \in \Upsilon_{\rm comp}^+} \sum_{k\in \mathbb{Z}^N} \int_{\mathrm{supp}\, \widehat{f}} \left\vert \widehat{f}(\xi + b_n^{-1}k) \right \vert^2 \frac{1}{\vert b_n \vert} \left\vert \widehat{\chi}_n(\xi) \right\vert^2 \mathrm{d}\xi < \infty.$$
This condition is given by the proof of Theorem~4 in \cite{Gilles2024frame}, replacing $\mathbb{Z}$ by $\mathbb{Z}^N$ and $\psi_n$ by $\chi_n$. It results from the fact that $\mathrm{supp}\, \widehat{f}$ and $\mathrm{supp}\, \widehat{\chi}_n$ are compact and that there are finitely many $\mathrm{supp}\, \widehat{\psi}_n$ that intersect $\mathrm{supp}\, \widehat{f}$. This proves the equivalence.
\end{proof}

\subsubsection{Non-compactly supported \texorpdfstring{$\boldsymbol{\widehat{\psi}}$}{Lg}}

In this section, we assume that the support of $\widehat{\psi}$ is not compact. The following theorems state sufficient conditions for which the system 
$$\{T_{b_nk}\psi_n\}_{(n,k)\in\Upsilon \times \Z^N} \quad \textrm{ and } \quad \{T_{b_nk}\chi_n\}_{(n,k)\in\Upsilon^+ \times \Z^N}$$ are frames.
The first theorem, for the system $\{\psi_n\}_{n\in\Upsilon}$, is a straight generalization of Theorem~8 in \cite{Gilles2024frame}
\begin{theorem}[Discrete frame]
\label{theo:nc_discret_frame}
If 
$$A =\inf_{\xi \in \R^N} \left( \sum_{n \in \Upsilon} \frac{1}{\vert b_n \vert} \left\vert \hat\psi_n(\xi) \right\vert^2 - \sum_{n \in \Upsilon} \sum_{k\neq0} \frac{1}{\vert b_n \vert} \left\vert \hat\psi_n(\xi) \widehat{\psi}_n(\xi-b_n^{-1}k) \right\vert\right) >0$$
and
$$B =\sup_{\xi \in \R^N} \sum_{n \in \Upsilon} \sum_{k\in \mathbb{Z}^N} \frac{1}{\vert b_n \vert} \left\vert \hat\psi_n(\xi) \widehat{\psi}_n(\xi-b_n^{-1}k) \right\vert<\infty,$$
then the system $\{T_{b_nk}\psi_n\}_{ n\in\Upsilon,k\in\mathbb{Z}^N}$ is a frame for $\mathrm{L}^2(\R^N)$ with frame bounds $A$ and $B$.
\end{theorem}

\begin{proof}
Let $f \in \mathcal{D}$.
We follow the proof of Theorem~8 in \cite{Gilles2024frame}, replacing $\mathbb{Z}$ by $\mathbb{Z}^N$ and $\R$ by $\R^N$.
First, we use arguments similar to the proof of Theorem~3.4 in \cite{labate2004approach}.
Since $\R^N$ can be written as a disjoint union $\R^N=\bigcup_{l\in\Z^N} b_n^{-1} (\mathbb{T}-l)$  with $\mathbb{T}=[0,1)^N$, we rewrite
\begin{align*}
\sum_{k\in\Z^N} \vert \langle f, T_{b_nk} \psi_n \rangle \vert^2 & = \sum_{k \in \Z^N} \left\vert \int_{\R^N} \widehat{f}(\xi) \overline{\widehat{\psi}_n(\xi)} e^{2\pi \imath ( b_n k \cdot \xi )} \rmd \xi \right\vert^2 \\
& = \sum_{k \in \Z^N} \left\vert \int_{b_n^{-1} \mathbb{T}} \left( \sum_{l\in\Z^N} \widehat{f}(\xi-b_n^{-1}l) \overline{\widehat{\psi}_n(\xi-b_n^{-1}l)} \right) e^{2\pi \imath b_n (k \cdot \xi )} \rmd \xi \right\vert^2 \\
& = \sum_{k \in \Z^N} \left\vert \int_{b_n^{-1} \mathbb{T}} \left( \sum_{l\in\Z^N} \widehat{f}(\xi-b_n^{-1}l) \overline{\widehat{\psi}_n(\xi-b_n^{-1}l)} \right) e^{2\pi \imath \frac{1}{b_n^{-1}} \left( k \cdot \xi \right)} \rmd \xi \right\vert^2,
\end{align*}
with $\mathbb{T}=[0,1)^N$.

The function $\xi \mapsto \sum_{l\in\Z^N} \widehat{f}(\xi-b_n^{-1}l) \overline{\widehat{\psi}_n(\xi-b_n^{-1}l)}$ belongs to $\Li(b_n^{-1} \mathbb{T})$ and each of its component is $b_n^{-1} \Z^N$-periodic. Thus, by Parseval's identity, we get
\begin{align*}
\sum_{k\in\Z^N} \vert \langle f, T_{b_nk} \psi_n \rangle \vert^2 &= \frac{1}{\vert b_n\vert} \int_{b_n^{-1} \mathbb{T}} \left\vert \sum_{l\in\Z^N} \widehat{f}(\xi-b_n^{-1}l) \overline{\widehat{\psi}_n(\xi-b_n^{-1}l)} \right\vert^2 \rmd \xi \\
&=\frac{1}{\vert b_n\vert} \int_{b_n^{-1} \mathbb{T}} \sum_{l\in\Z^N} \widehat{f}(\xi-b_n^{-1}l) \overline{\widehat{\psi}_n(\xi-b_n^{-1}l)} \overline{\sum_{u\in\Z^N} \widehat{f}(\xi-b_n^{-1}u) \overline{\widehat{\psi}_n(\xi-b_n^{-1}u)} } \rmd \xi. 
\end{align*}
Hence, by the change of indices $u=l+k$, we get
\begin{align*}
\sum_{k\in\Z^N} \vert \langle f, & T_{b_nk} \psi_n \rangle \vert^2 \\
  &=\frac{1}{\vert b_n\vert} \int_{b_n^{-1} \mathbb{T}} \sum_{l,k\in\Z^N} \widehat{f}(\xi-b_n^{-1}l) \overline{\widehat{\psi}_n(\xi-b_n^{-1}l)} \, \overline{\widehat{f}(\xi-b_n^{-1}(l+k))} \widehat{\psi}_n(\xi-b_n^{-1}(l+k)) \rmd \xi \\
 &=\frac{1}{\vert b_n\vert} \sum_{l,k\in\Z^N}  \int_{b_n^{-1} \mathbb{T}}\widehat{f}(\xi-b_n^{-1}l) \overline{\widehat{\psi}_n(\xi-b_n^{-1}l)} \, \overline{\widehat{f}(\xi-b_n^{-1}(l+k))} \widehat{\psi}_n(\xi-b_n^{-1}(l+k)) \rmd \xi \\
  &=\frac{1}{\vert b_n\vert} \sum_{k\in\Z^N}  \int_{\R^N}\widehat{f}(\xi) \overline{\widehat{\psi}_n(\xi)} \, \overline{\widehat{f}(\xi-b_n^{-1}k)} \widehat{\psi}_n(\xi-b_n^{-1}k) \rmd \xi.
\end{align*}
Splitting the terms when $k=0$ and $k\neq 0$, we obtain
\begin{equation}\label{eq:split_equality_frame}
\sum_{n\in\Upsilon} \sum_{k\in\Z^N} \vert \langle f, T_{b_nk} \psi_n \rangle \vert^2 =  \int_{\R^N} \left\vert \widehat{f}(\xi) \right\vert^2 \sum_{n\in\Upsilon} \frac{1}{\vert b_n\vert} \left\vert \widehat{\psi}_n(\xi) \right\vert^2 \rmd \xi + R(f),
\end{equation}
where
\begin{equation*}
R(f) =\sum_{n\in\Upsilon} \sum_{k\neq 0} \frac{1}{\vert b_n\vert} \int_{\R^N}\widehat{f}(\xi) \overline{\widehat{f}(\xi-b_n^{-1}k)} \, \overline{\widehat{\psi}_n(\xi)} \widehat{\psi}_n(\xi-b_n^{-1}k) \rmd \xi.
\end{equation*}

Finally, the arguments of the proof of Theorem~3.1 in \cite{christensen2008frame}, with $d=N$, $C_j=b_n$, $g_j=\psi_n$ and $\mathcal{J}=\Upsilon$, give the results.
\end{proof}

\begin{theorem}[Discrete symmetric frame]
If 
$$A =\inf_{\xi \in \R^N} \left( \sum_{n \in \Upsilon^+} \frac{1}{\vert b_n \vert} \left\vert \widehat{\chi}_n(\xi) \right\vert^2 - \sum_{n \in \Upsilon^+} \sum_{k\neq0} \frac{1}{\vert b_n \vert} \left\vert \widehat{\chi}_n(\xi) \widehat{\chi}_n(\xi-b_n^{-1}k) \right\vert\right) >0$$
and
$$B =\sup_{\xi \in \R^N} \sum_{n \in \Upsilon^+} \sum_{k\in \mathbb{Z}^N} \frac{1}{\vert b_n \vert} \left\vert \widehat{\chi}_n(\xi) \widehat{\chi}_n(\xi-b_n^{-1}k) \right\vert<\infty,$$
then the system $\{T_{b_nk}\chi_n\}_{ n\in\Upsilon^+,k\in\mathbb{Z}^N}$ is a frame for $\mathrm{L}^2(\R^N)$ with frame bounds $A$ and $B$.
\end{theorem}

\begin{proof}
This results from the argumentation of the proof of \Cref{theo:nc_discret_frame} by replacing $\Upsilon$ by $\Upsilon^+$ and $\psi_n$ by $\chi_n$.
\end{proof}

\section{Empirical wavelet systems from affine deformations}
\label{sec:aff_def}

In this section, we consider a partition $\{\Omega_n\}_{n\in\Upsilon}$ for which we can find diffeomorphisms $\gamma_n$ that are affine functions, i.e., of the form $\gamma_n(\xi)=\beta_n(\xi-\eta_n)$ with $\beta_n$ a linear function.
If the construction of the empirical wavelet systems introduced in \Cref{sec:ews} is natural in the Fourier domain, it is also possible to build them directly in the spatial domain for affine mappings. 

First, the following Lemma gives the explicit expression of a function with Fourier transform deformed by a linear function.

\begin{lemma}\label{lem:fourdeform}
Let $\beta : x \mapsto Ax$ be a linear function (non identically zero), with $A \in \R^{N\times N}$.
The function $g$ corresponding to the inverse Fourier transform of the deformation of the Fourier transform $\widehat{f}$ of a function $f$ by $\beta$, i.e., $\widehat{g}=\widehat{f}\circ\beta$, is given by, for every $x \in \R^N$,
  \begin{equation}\label{eq:invformula}
  g(x)=\frac{1}{\left|det\; A \right|} f(A^{-\intercal}x),
\end{equation}
 where $\beta^{-\intercal}=(\beta^{-1})^\intercal$.
\end{lemma}
\begin{proof}
  By taking the inverse Fourier transform, we have (using the substitution $\xi=\beta^{-1}(u) \to d\xi=|\mathrm{det}\; J_{\beta^{-1}}(u)|du$)
\begin{align*}
g(x)= \int_{\R^N} \widehat{f}(\beta(\xi))e^{2\pi \imath (\xi\cdot x)}\mathrm{d}\xi&=
\int_{\beta(\R^N)}\widehat{f}(u)e^{2\pi\imath(\beta^{-1}(u)\cdot x )} \left|\mathrm{det}\; J_{\beta^{-1}}(u) \right|\mathrm{d}u\\
&=\frac{1}{\left|det\; A \right|} \int_{\R^N}\widehat{f}(u)e^{2\pi \imath ( u\cdot \beta^{-\intercal}(x) )} \mathrm{d}u\\
&=\frac{1}{\left|det\; A \right|} \F^{-1}\left(\widehat{f}\right)(\beta^{-\intercal}(x))\\
&= \frac{1}{\left|det\; A \right|} f(A^{-\intercal}x).
\end{align*}
\end{proof}

Finally, for affine diffeormorphisms, empirical wavelet systems $\{\psi_n\}_{n\in\Upsilon}$ can be built in the spatial domain using the following proposition.

\begin{proposition}[Spatial domain construction]
\label{prop:spatial_domain_construction}
  Let $\{\Omega_n\}_{n\in\Upsilon}$ be a partition of the Fourier domain and $\psi$ a wavelet kernel. Let assume there exists a set of linear functions $\{\beta_n: \xi \mapsto A_n \xi\}_{n\in\Upsilon}$ such that
  $\Lambda=\beta_n(\Omega_n-\eta_n)$ if $\Omega_n$ is bounded and $\Lambda \subsetneq \beta_n(\Omega_n-\eta_n)$ otherwise.
The set of empirical wavelets, $\{\psi_n\}_{n\in\Upsilon}$, defined in the Fourier domain by, for every $\xi \in \R^N$, 
$$\widehat{\psi}_n(\xi) = \sqrt{\left\vert \mathrm{det} J_{\beta_n} (\xi - \eta_n) \right\vert } \, \widehat{\psi}( \beta_n(\xi-\eta_n)) =\sqrt{\left\vert \mathrm{det} A_n \right\vert } \, \widehat{\psi}( A_n(\xi-\eta_n)),$$
is given in the spatial domain by, for every $x\in\R^N$,
  \begin{equation*}
  \psi_n(x)=\frac{1}{\sqrt{\left\vert \mathrm{det} A_n \right\vert }}\psi(A_n^{-\intercal}x) \, e^{2\pi \imath \, \eta_n x}.
  \end{equation*}
\end{proposition}

\begin{proof}
Moreover, since the translation by $\eta_n$ in the Fourier domain is a modulation of frequency $\eta_n$ in the spatial domain, we can rewrite
\begin{equation*}
\psi_n(x) = \F^{-1} \left( \sqrt{\left\vert \mathrm{det} A_n \right\vert } \widehat{\psi} \circ\beta_n \right)(x) \, e^{2\pi\imath \eta_n x}.
\end{equation*}
By applying \Cref{lem:fourdeform} with $\widehat{f}= \sqrt{\left\vert \mathrm{det} A_n \right\vert } \widehat{\psi}$
and $\beta=\beta_n$, we obtain the result.
\end{proof}

\begin{example}\label{ex:1D_spatial_domain}
As in \Cref{ex:psin_1D}, we consider the 1D case with $\widehat{\psi}$ of support $\Lambda$ the open interval of size $1$ centered on $0$ and $\{\Omega_n\}_{n\in\Upsilon}$ a partition of open bounded intervals with center $\omega_n$. The diffeomorphism $\gamma_n: \xi \mapsto \frac{1}{\vert \Omega_n \vert}(\xi-\omega_n)$, satisfying $\Lambda=\gamma_n(\Omega_n)$, can be rewritten 
$\gamma_n(\xi)=\beta_n(\xi-\eta_n)$ with $\beta_n: \xi \mapsto \xi/\vert \Omega_n \vert$ and $\eta_n =\omega_n$. The associated empirical wavelet system is therefore given in the spatial domain by, for every $x\in\R$,
\begin{equation*}
\begin{aligned}
\psi_n(x) 
&= \sqrt{ \vert \Omega_n \vert } \, \psi (\vert \Omega_n \vert x ) \, e^{2\pi \imath \, \omega_n x}.
\end{aligned}
\end{equation*}
We thus retrieve the definition given in \cite{cewt}.
\end{example}

\begin{example}\label{ex:2D_classic_wavelet}
We can also build wavelets on a 2D dyadic tiling. We consider a separable wavelet kernel $\widehat{\psi}^{\rm 2D}(\xi_1,\xi_2)=\widehat{\psi}(\xi_1)\widehat{\psi}(\xi_2)$, where $\widehat{\psi}$ is a scaling function supported by an open interval $\Lambda$ centered in $0$, 
and the diffeomorphism $\gamma_{j,n}: \xi \mapsto 2^j(\xi - \omega_n)$, $n=1,2,3$, with $\omega_1=(\omega_0,0)$, $\omega_2=(0,\omega_0)$ and $\omega_3=(\omega_0,\omega_0)$.
The associated translated empirical wavelet system is given in the spatial domain by, for every $(x,y)\in\R^2$ and $(l,m)\in\Z^2$,
\begin{equation*}
\begin{aligned}
\psi_{j,n}((x,y) - 2^j(l,m)) &= \frac{1}{2^j} \, \psi \left(\frac{x - 2^jl}{2^j}\right) \psi \left(\frac{y  - 2^jm}{2^j}\right).
\end{aligned}
\end{equation*}
Finally, the transform resulting from \cref{def:discrete_transf} using this system is similar to the classic discrete wavelet transform.
\end{example}

The construction of a symmetric empirical wavelet system $\{\chi_n\}_{n\in\Upsilon}$, induced by affine diffeomorphisms, in the spatial domain stems from \Cref{prop:spatial_domain_construction}, using \Cref{eq:spatial_domain_symwavelet}.

\begin{example}
We consider again the 1D domain. From the system $\{\psi_n\}_{n\in\Upsilon}$ defined as in \Cref{ex:1D_spatial_domain}, we can build a symmetric system $\{\chi_n\}_{n\in\Upsilon}$ as follows, for every $x\in\R$,
$$\chi_0(x)  = \sqrt{ \vert \Omega_0 \vert } \, \psi (\vert \Omega_0 \vert x ),$$
and, for $n \neq 0$,
\begin{equation*}
\begin{aligned}
\chi_n(x)  &= \sqrt{ \frac{\vert \Omega_n \vert}{2} } \psi (\vert \Omega_n \vert x )e^{2\pi \imath \, \omega_n x} + \sqrt{ \frac{\vert \Omega_n \vert}{2} } \psi (\vert \Omega_n \vert x )e^{-2\pi \imath \, \omega_n x}   \\
&= \sqrt{2 \vert \Omega_n \vert } \, \psi (\vert \Omega_n \vert x ) \,\cos(2\pi\omega_n x).
\end{aligned}
\end{equation*}
In particular, this extends the definitions of 1D empirical wavelet systems in \cite{cewt} to symmetric partitions of the Fourier domain.
\end{example}

\section{Construction of empirical wavelet systems}
\label{sec:num_exp}

In this section, 2D continuous empirical wavelet systems are constructed from Gabor and Shannon wavelet kernels and we examine their guarantees of reconstruction. Their practical behaviors are analyzed through numerical experiments conducted on two different (real-valued) images: a toy image of size $256 \times 256$ and the classic Barbara image of size $512 \times 512$. These images are shown in \Cref{fig:images} (left).

For the Fourier domain partitioning of both images, the $N_m$ harmonic modes are detected by the scale-space representation \cite{Gilles2014a} on the logarithm of the Fourier spectrum, using a scale-space parameter set to $s_0=0.8$. We get $N_m=10$ for the toy image and $N_m=15$ for the Barbara image. The Fourier domain is then partitioned by separating the detected modes using either the Watershed \cite{hurat2020empirical} or Voronoi \cite{gilles2022empirical} methods, which provide as many connected supports of low constrained shapes as modes.
\Cref{fig:images} (middle and right) shows the (symmetric) Fourier tranform of both images, partitioned by either Watershed or Voronoi into $N_m$ symmetric regions $\Omega_n \cup \Omega_{-n}$, with $n\in\Upsilon^+=\{0,\ldots,N_m-1\}$.

\begin{figure}[t]
\centerline{
\subfloat{
  \includegraphics[width=.3\textwidth]{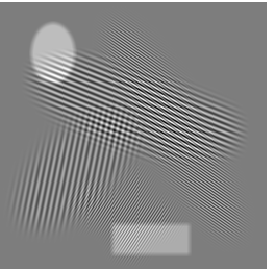}
  \includegraphics[width=.3\textwidth]{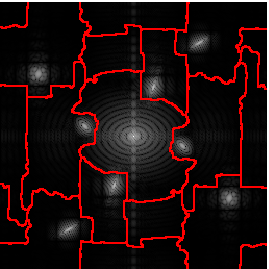}
  \includegraphics[width=.3\textwidth]{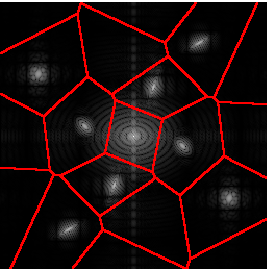}
  } }
  \centerline{
\subfloat{
    \includegraphics[width=.3\textwidth]{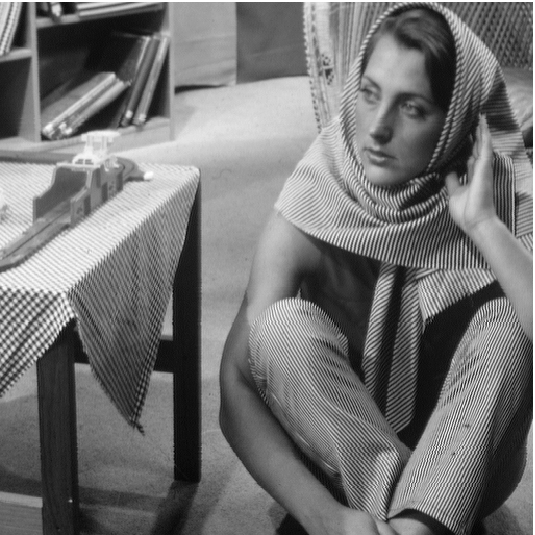}
  \includegraphics[width=.3\textwidth]{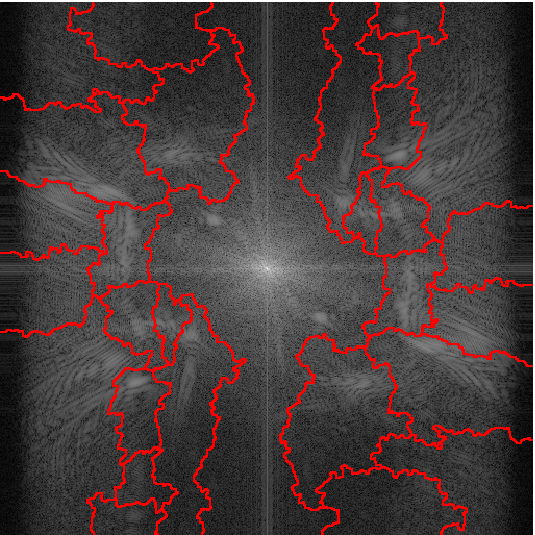}
  \includegraphics[width=.3\textwidth]{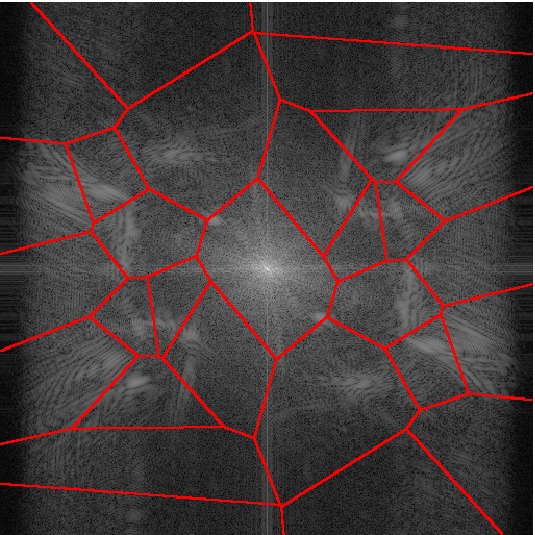}
  } }
 \caption{{\bf Images and Fourier partitions.} (Top) Toy image of size $256 \times 256$ and (bottom) classic Barbara image of size $512 \times 512$, along with the (middle) Watershed and (right) Voronoi partitions (overlapping in red) of the logarithm of their Fourier spectra.}
  \label{fig:images}
\end{figure}

To numerically construct empirical wavelet systems $\{\chi_n\}_{n\in\Upsilon^+}$ as in \Cref{eq:symwavelet_def}, we need to compute an estimate $\breve \gamma_n$ of the diffeomorphism $\gamma_n$.
%This estimation is performed using the Demons algorithm \cite{thirion1998image}.
This estimation is performed using the demons algorithm \cite{thirion1998image}, which is inspired by a diffusion process and is widely used in medical image registration.
This algorithm estimates a displacement field representing the desired mapping by alternating between solving the flow equations and regularization.
For each region $\Omega_n$, the pair of smoothing parameter and number of multiresolution image pyramid levels is selected by a grid search minimizing the quadratic risk $\Vert \Lambda - \breve \gamma_n(\Omega_n)\Vert_2^2$ on the values in $(0.3,0.35,\ldots,0.7)\times (n_P-2,n_P-1,n_P)$, where $n_P$ is the highest integer such that $2^{n_P}$ is smaller than each dimension of the image.
The number of iterations at the $n_P$ pyramid levels are set to $(2^4,\ldots,2^{n_P+1})$ from the highest to the lowest pyramid level.

For all the numerical experiments, the symmetric empirical wavelet transform is computed as in \Cref{def:discrete_transf} with $b_n=1$ for every $n\in\Upsilon^+$, so that the results for continuous wavelet frames in \Cref{sec:ewf} apply. 
To each wavelet transform coefficient $\E_{\chi}^f(\cdot,n)$ of an image $f$ of size $M \times M$, we associate the wavelet transform spectrum defined as $\vert \E_{\chi}^f(\cdot,n) \vert^2$
and the Fourier spectrum energy $\mathrm{E}_n$ of the underlying region $\Omega_n \cup \Omega_{-n}$ that reads
 $$\mathrm{E}_n = \frac{1}{M^2} \sum_{(m_1,m_2) \in \Omega_n \cup \Omega_{-n}} \left\vert \widehat{f}(m_1,m_2) \right\vert^2.$$  
The reconstruction $r$ of $f$ given by \Cref{eq:symreconstruction} is assessed by the Mean Squared Error (MSE) given by
$$ \mathrm{MSE}(r) = \frac{1}{M^2} \sum_{m_1=1}^M \sum_{m_2=1}^M \vert r(m_1,m_2) - f(m_1,m_2) \vert^2.$$
%\left\Vert \tilde f - f \right\Vert_{\rm F}^2.$$

\subsection{2D empirical Gabor wavelets}

The Fourier transform of the 1D Gabor wavelet kernel is given by (see \cite{cewt}), for every $u \in \R$,
$$\widehat{\psi}^\mathrm{1D-G}(u)=e^{-\pi\left(\frac{5}{2}u\right)^2},$$
which is mostly supported by $(-\frac{1}{2},\frac{1}{2})$. We define its extension to 2D by, for every $u \in \R^2$,
\begin{equation*}
\widehat{\psi}^\mathrm{G}(u)=e^{-\pi \left(\frac{5}{2} \right)^2\Vert u \Vert^2_2},
\end{equation*}
which is mostly supported by the open disk of center $0$ and radius $1/2$ denoted $\Lambda = \mathrm{B}_2(0,1/2)$.

The following proposition gives guarantees of reconstruction from a Gabor empirical wavelet systems.

\begin{proposition}[Gabor empirical wavelet reconstruction]
Let assume that the diffeomorphisms $\gamma_n$ satisfy, for $a.e. \, \xi \in \R^2$, $\vert \{ m \in \Upsilon \mid \gamma_m(\xi)=\gamma_n(\xi) \} \vert \leq K_\xi \in \mathbb{N}$ for every $n\in\Upsilon$ and $\{ \vert \mathrm{det}\; J_{\gamma_n}(\xi) \vert \}_{n \in\Upsilon}$ is a bounded sequence. Then, the continuous reconstruction is guaranteed for the empirical wavelet systems $\{\psi_n\}_{n\in\Upsilon}$ and $\{\chi_n\}_{n\in\Upsilon}$ induced by $\widehat{\psi}^\mathrm{G} \circ \gamma_n$, and is given by \Cref{eq:reconstruction,eq:symreconstruction}.
\end{proposition}

\begin{proof}
We have, for every $\xi \in \R^2$, 
\begin{align*}
\sum_{n\in\Upsilon}\left|\widehat{\psi}_n(\xi)\right|^2 = \sum_{n\in\Upsilon}\vert \mathrm{det}\; J_{\gamma_n}(\xi) \vert \left|\widehat{\psi}(\gamma_n(\xi))\right|^2 = \sum_{n\in\Upsilon} \vert \mathrm{det}\; J_{\gamma_n}(\xi) \vert e^{-2\pi(\frac{5}{2})^2\Vert \gamma_n(\xi) \Vert^2}.
\end{align*}

First, since $\gamma_n$ is a diffeomorphism, we have, for every $\xi \in \R^2$, $\vert \mathrm{det}\; J_{\gamma_n}(\xi) \vert > 0 $ and therefore
\begin{equation*}
 \sum_{n\in\Upsilon}\left|\widehat{\psi}_n(\xi)\right|^2 > 0.
 \end{equation*}
 
Now, let us denote $\Gamma_\xi = \left\{ m \in \Upsilon \mid \gamma_m(\xi) = \gamma_n(\xi) \Rightarrow n \geq m \right\} $. For $a.e. \, \xi \in \R^2$  and every $n\in\Gamma_\xi$, the $\gamma_n(\xi)$ are all different and the condition $\vert \{ m \in \Upsilon \mid \gamma_m(\xi)=\gamma_n(\xi) \} \vert \leq K_\xi \in \mathbb{N}$ means that there are at most $K_\xi$ integers $m \in \Upsilon$ for which the $\gamma_m(\xi)$ have the same value as $\gamma_n(\xi)$. 
It follows that, for $ a.e. \, \xi \in \R^2$,
\begin{align*}
\sum_{n\in\Upsilon}\left|\widehat{\psi}_n(\xi)\right|^2 & \leq \max_{n \in \Upsilon} \vert \mathrm{det}\; J_{\gamma_n}(\xi) \vert \sum_{n\in\Upsilon} e^{-2\pi(\frac{5}{2})^2\Vert \gamma_n(\xi) \Vert^2}\\ 
& \leq K_\xi \max_{n \in \Upsilon} \vert \mathrm{det}\; J_{\gamma_n}(\xi) \vert \sum_{n\in\Gamma_\xi} e^{-2\pi(\frac{5}{2})^2\Vert \gamma_n(\xi) \Vert^2} \\
& \leq K_\xi \max_{n \in \Upsilon} \vert \mathrm{det}\; J_{\gamma_n}(\xi) \vert \int_{\R^2} e^{-2\pi(\frac{5}{2})^2\Vert u \Vert^2} \mathrm{d}u < \infty,
\end{align*}
where the last line comes from the fact that $\{\gamma_n(\xi) | n\in\Gamma_\xi \}$ is a sampling of $\R^2$ and $u \mapsto e^{-2\pi(\frac{5}{2})^2\Vert u \Vert^2}$ is a positive function.
Then, \Cref{theo:icewt} applies and gives a dual frame of $\{\psi_n\}_{n\in\Upsilon}$. This guarantees the reconstruction using \Cref{eq:reconstruction}.

In addition, in the case of a symmetric partition $\{\Omega_n\}_{n\in\Upsilon}$, we show similarly that, for $ a.e. \, \xi \in \R^2$,
\begin{align*}
\sum_{n\in\Upsilon\setminus\{0\}}\left|\widehat{\psi}_n(\xi)\right| \left|\widehat{\psi}_{-n}(\xi)\right| & =\sum_{n\in\Upsilon\setminus\{0\}}\sqrt{\vert \mathrm{det}\; J_{\gamma_n}(\xi) \mathrm{det}\; J_{\gamma_{n}}(-\xi) \vert} e^{-\pi(\frac{5}{2})^2 \left(\Vert \gamma_{n}(\xi) \Vert^2+\Vert \gamma_{-n}(\xi) \Vert^2 \right)} \\
&=K_\xi \max_{n\in\Upsilon\setminus\{0\}}\vert \mathrm{det}\; J_{\gamma_n}(\xi) \vert \sum_{n\in\Gamma_\xi} e^{-\pi(\frac{5}{2})^2 \left(\Vert \gamma_{n}(\xi) \Vert^2+\Vert \gamma_{-n}(\xi) \Vert^2 \right)} \\
&=K_\xi \max_{n\in\Upsilon\setminus\{0\}}\vert \mathrm{det}\; J_{\gamma_n}(\xi) \vert \int_{\R^4} e^{-\pi(\frac{5}{2})^2 \Vert u \Vert^2 } \rmd u < \infty. 
\end{align*}
Then, \Cref{theo:symicewt} applies and gives a dual frame of $\{\chi_n\}_{n\in\Upsilon^+}$. The reconstruction is then given by \Cref{eq:symreconstruction}.
\end{proof}

For the toy image, \Cref{fig:toy_image_watershed_gabor_transform,fig:toy_image_voronoi_gabor_transform} compare the preimage of $\Lambda = \mathrm{B}_2(0,\frac{1}{2})$ under the diffeomorphism estimate $\breve \gamma_n$ to the targeted $\Omega_n$ and show the symmetric empirical Gabor wavelet coefficients induced by $\widehat{\psi}^\mathrm{G} \circ \breve \gamma_n$, for the Watershed and Voronoi partitions, respectively, and $n\in\{0,\ldots,9\}$. 
The diffeomorphisms $\gamma_n$ are well estimated for both partitions but the more constrained shapes of the Voronoi partition allow for better estimates. The MSE of the reconstruction from the Watershed and Voronoi transforms are, respectively, $6.93 \times 10^{-30}$ and $1.65 \times 10^{-31}$, %$9.63 \times 10^{-18 }$ and $1.61 \times 10^{-18}$, 
confirming the accurate diffeomorphism estimation. In addition, the different components of the toy image are well recovered by the wavelet coefficients associated with the highest Fourier spectrum energies.

\begin{figure}[!ht]
\centerline{
  \includegraphics[width=\textwidth]{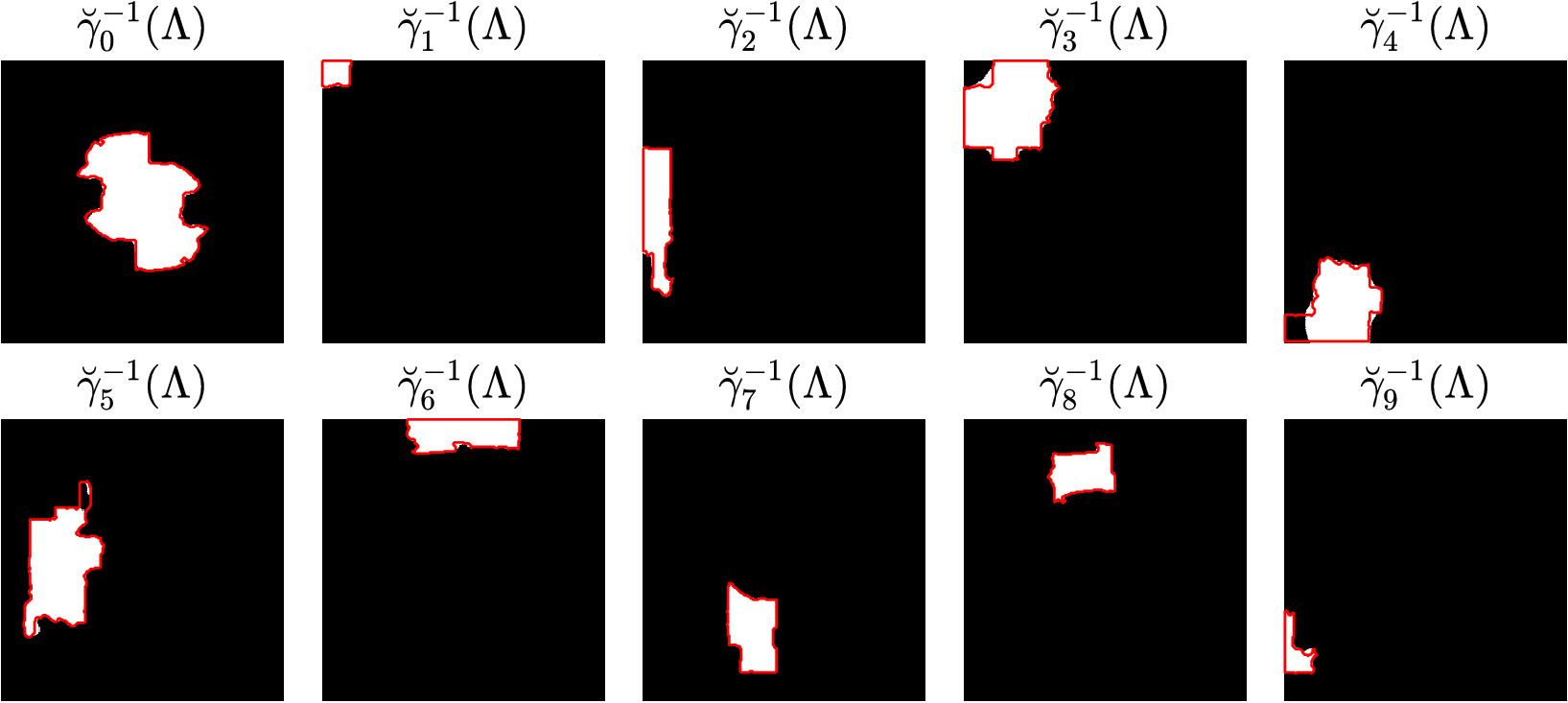} }
  \noindent\rule[5pt]{\textwidth}{0.4pt}
  \centerline{ \includegraphics[width=\textwidth]{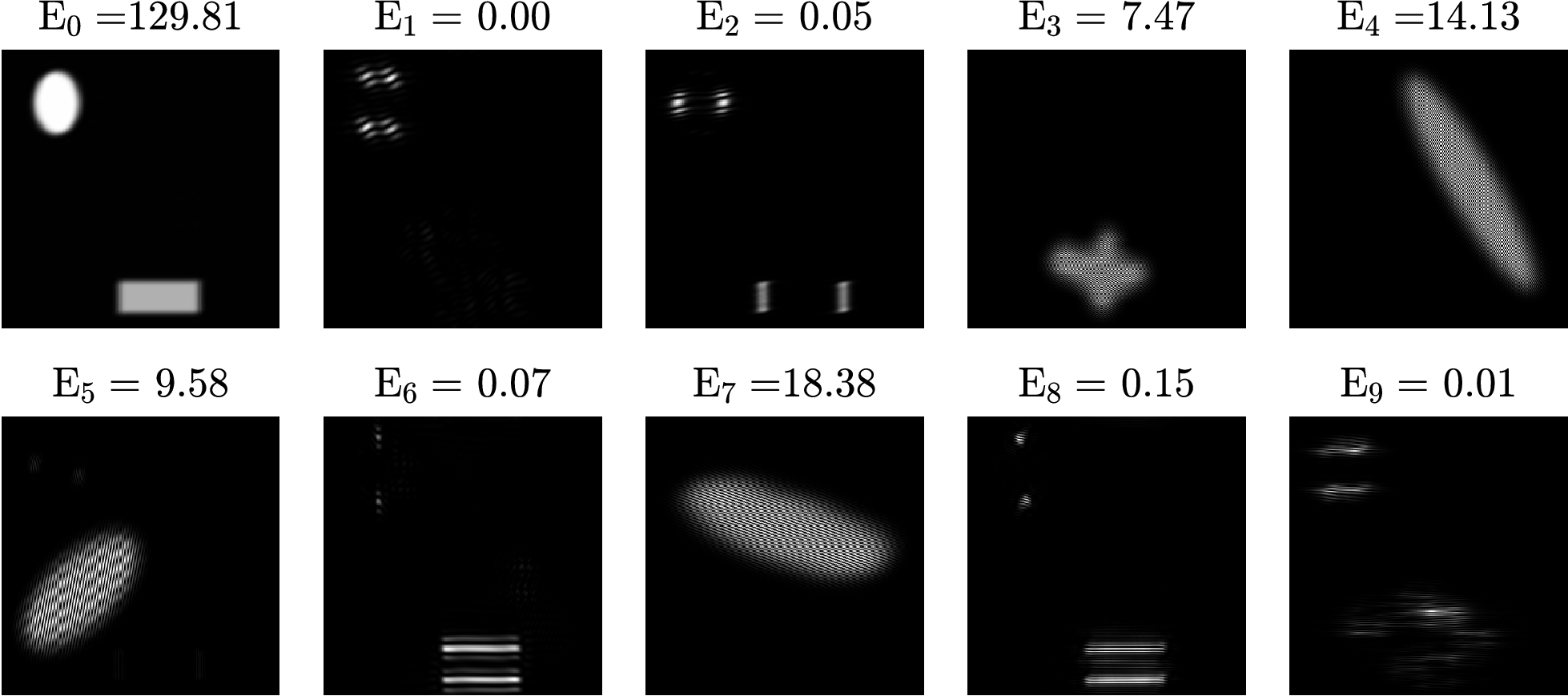} }
 \caption{ {\bf Gabor Watershed transform of the toy image.} (Top) Sets $\breve \gamma_n^{-1}(\Lambda)$ (in white) for the diffeomorphisms $\gamma_n$ mapping the Watershed regions $\Omega_n$ (with contour in red) to the disk $\Lambda = \mathrm{B}_2(0,1/2) $ and (bottom) resulting empirical Gabor wavelet transform spectra, for the toy image and $n \in \{0,\ldots,9\}$. The Fourier spectrum energies $\mathrm{E}_n$ over the regions $\Omega_n \cup \Omega_{-n}$ are indicated above the wavelet spectra. }
  \label{fig:toy_image_watershed_gabor_transform}
\end{figure}

\begin{figure}[!ht]
\centerline{
  \includegraphics[width=\textwidth]{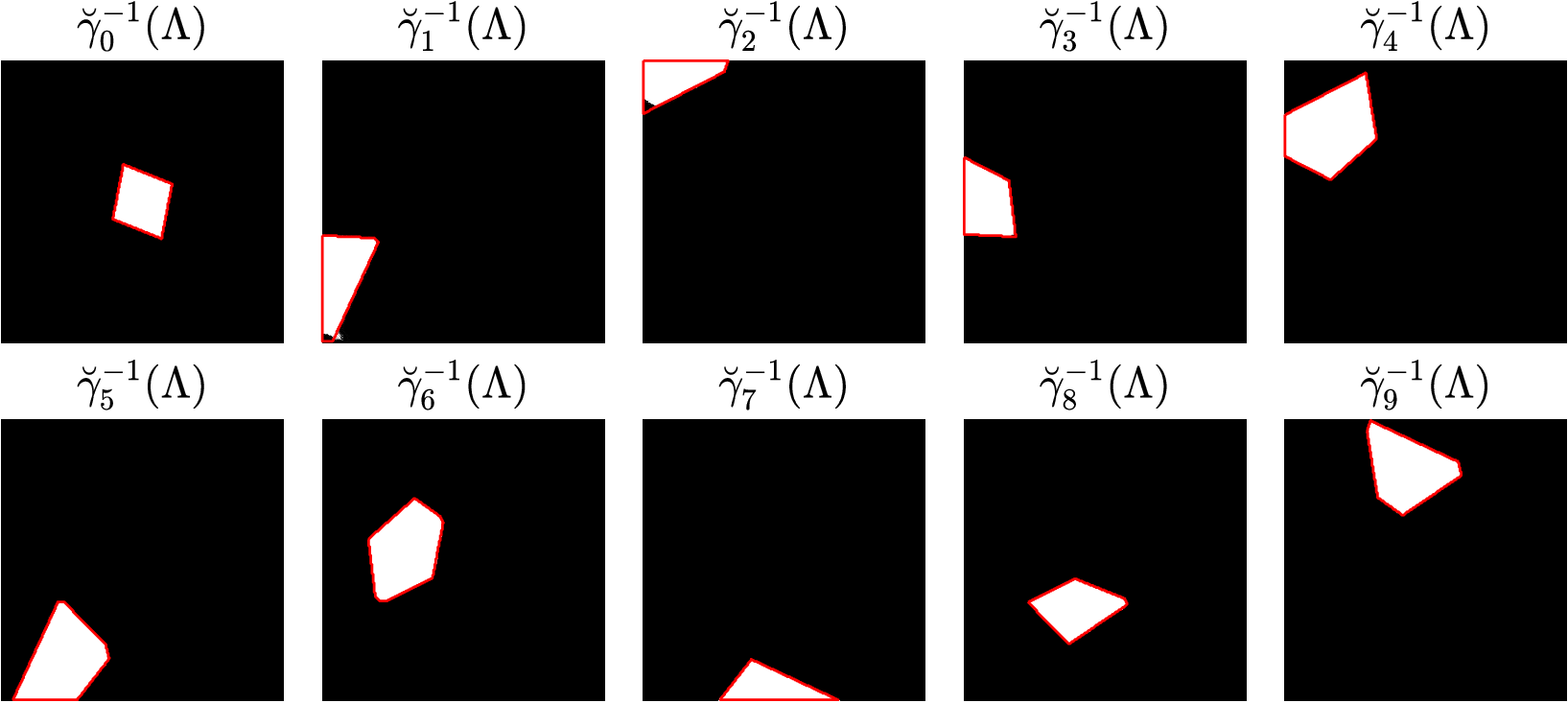} }
  \noindent\rule[5pt]{\textwidth}{0.4pt}
  
  \centerline{ \includegraphics[width=\textwidth]{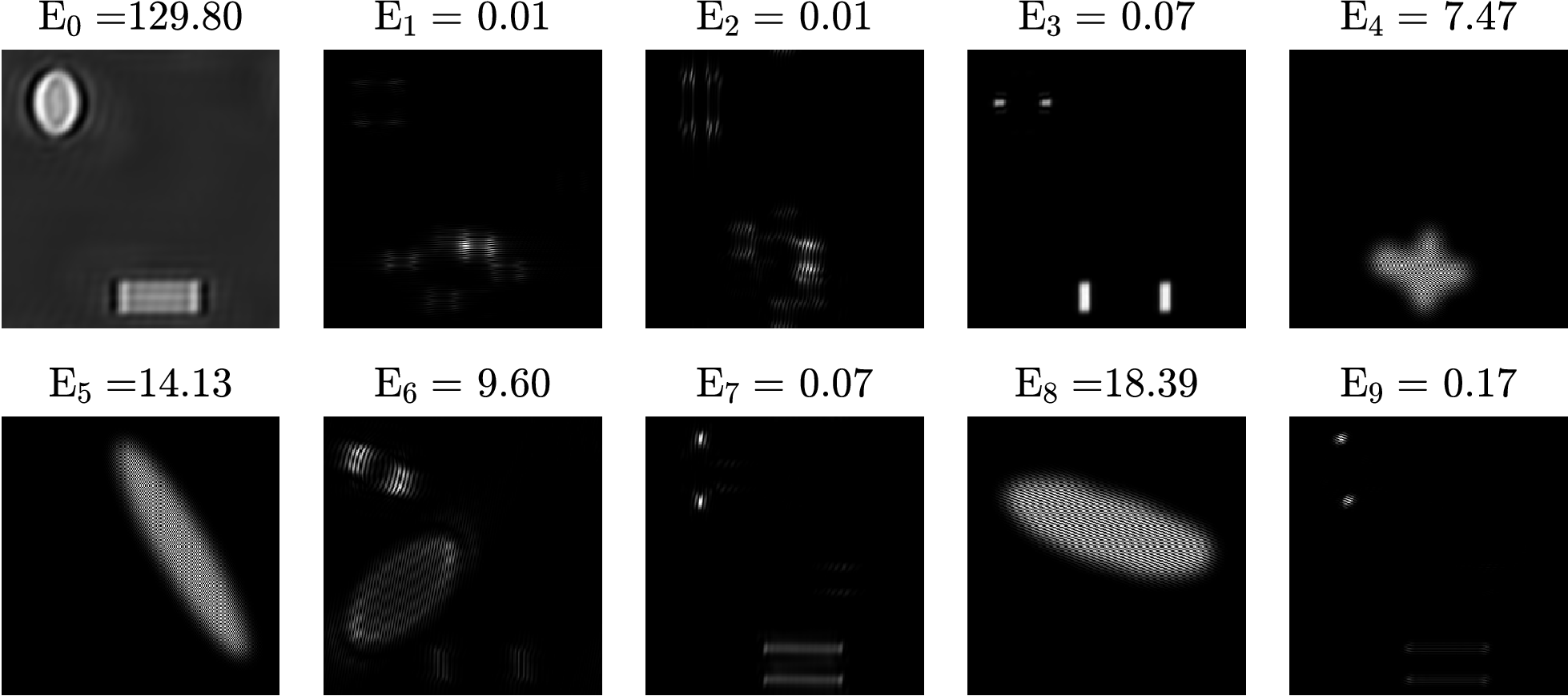} }
 \caption{ {\bf Gabor Voronoi transform of the toy image.} (Top) Sets $\breve \gamma_n^{-1}(\Lambda)$ (in white) for the diffeomorphisms $\gamma_n$ mapping the Voronoi regions $\Omega_n$ (with contour in red) to the disk $\Lambda = \mathrm{B}_2(0,1/2) $ and (bottom) resulting empirical Gabor wavelet transform spectra, for the toy image and $n \in \{0,\ldots,9\}$. The Fourier spectrum energies $\mathrm{E}_n$ over the regions $\Omega_n \cup \Omega_{-n}$ are indicated above the wavelet spectra. }
  \label{fig:toy_image_voronoi_gabor_transform}
\end{figure}

For the Barbara image, 
\Cref{fig:barbara_watershed_gabor_transform,fig:barbara_voronoi_gabor_transform} compare the preimage of $\Lambda = \mathrm{B}_2(0,\frac{1}{2})$ under $\breve \gamma_n$ to the targeted $\Omega_n$ and show the symmetric empirical Gabor wavelet coefficients induced by $\widehat{\psi}^\mathrm{G} \circ \breve \gamma_n$, for the Watershed and Voronoi partitions, respectively, and $n\in\{0,\ldots,9\}$. The diffeomorphism estimation is much more accurate on the Vornoi partition than on the Watershed partition, in particular for the region $\Omega_{0}$ which is very large for the Watershed partition. Thus, in terms of reconstruction, Watershed and Voronoi partition lead to MSE of $5.02 \times 10^{-25}$ and $2.10 \times 10^{-30}$, 
respectively.

\begin{figure}[htbp]
\centerline{
  \includegraphics[width=\textwidth]{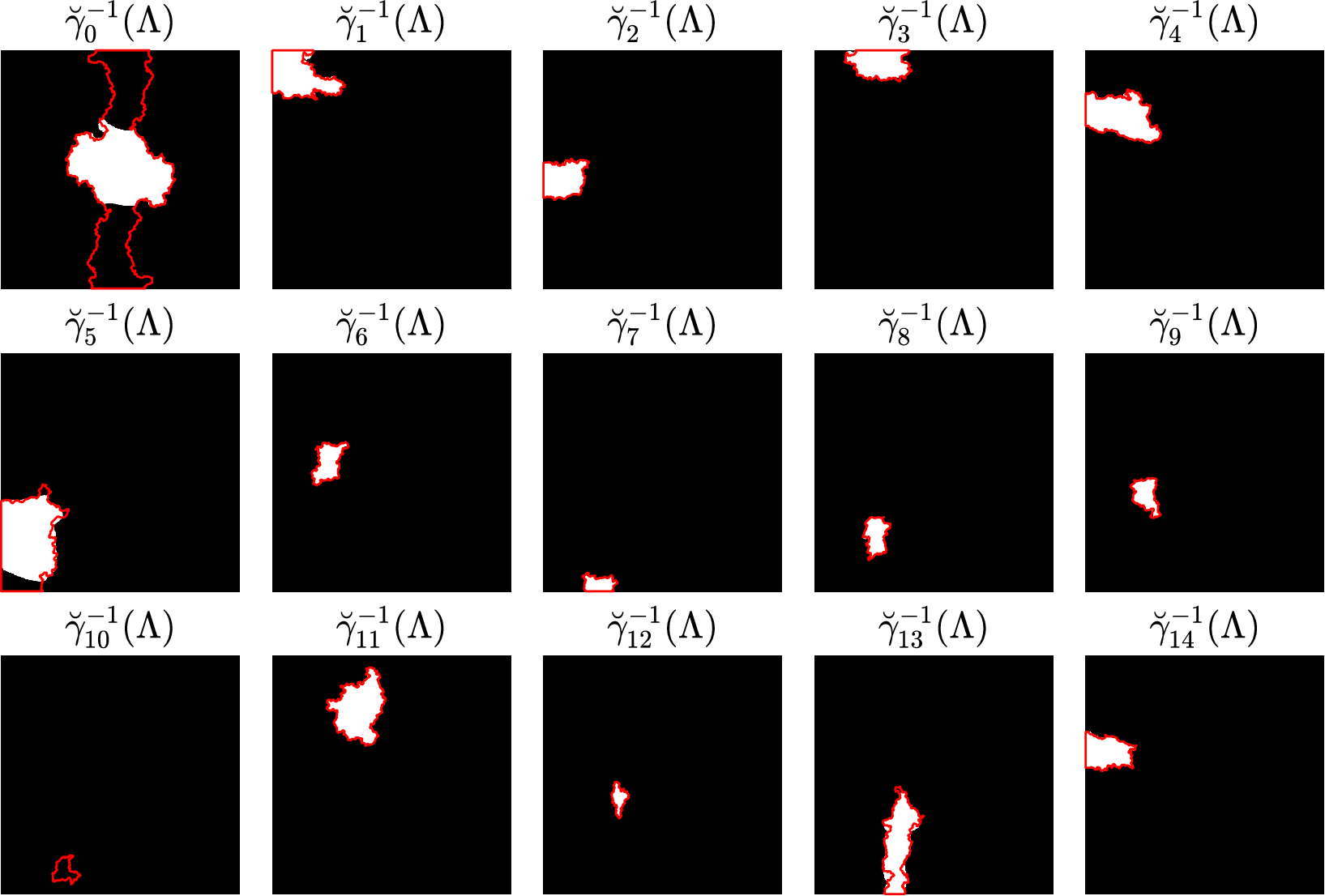} }
  \noindent\rule[5pt]{\textwidth}{0.4pt}
  \centerline{ \includegraphics[width=\textwidth]{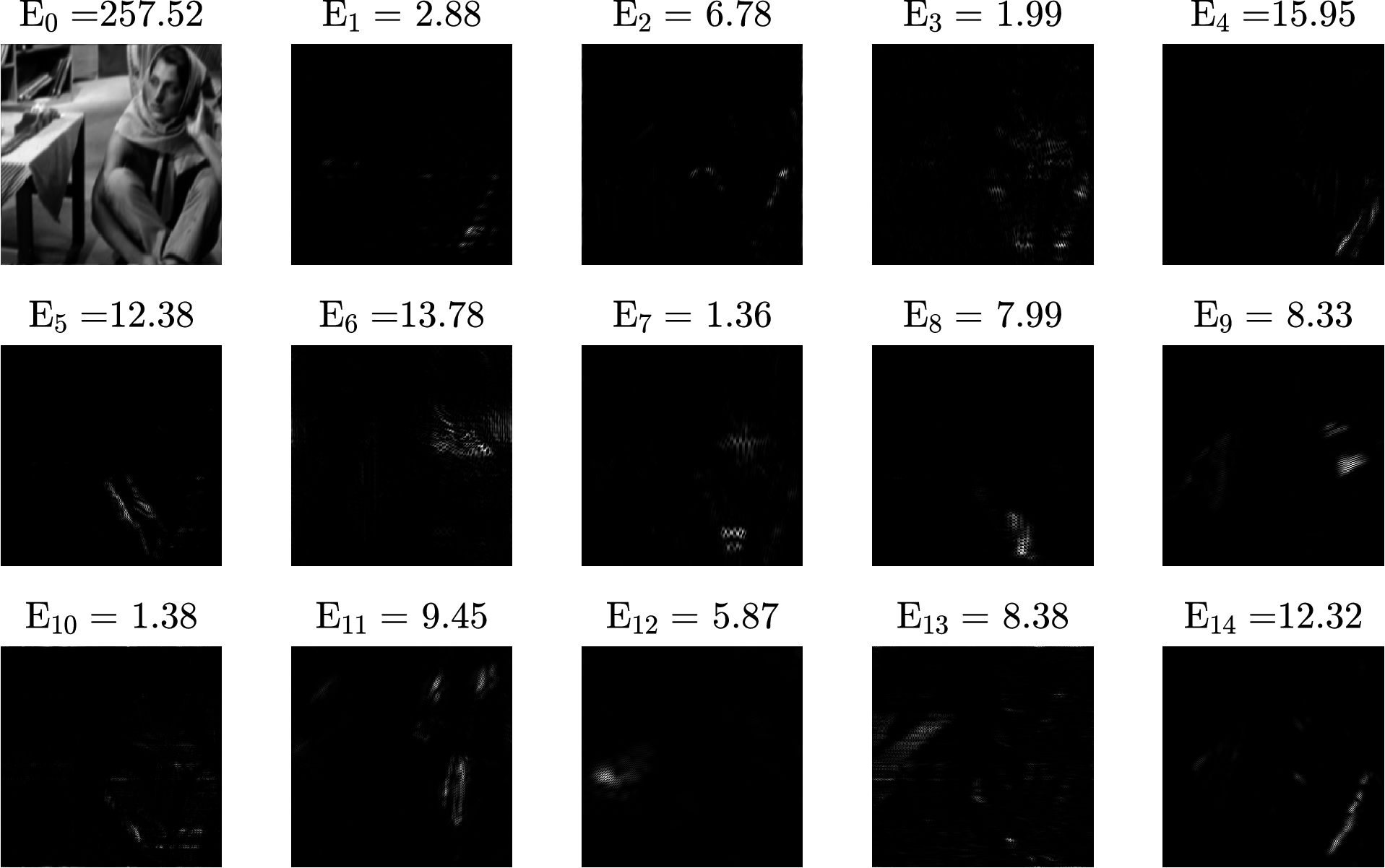} }
 \caption{ {\bf Gabor Watershed transform of the Barbara image.} (Top) Sets $\breve \gamma_n^{-1}(\Lambda)$ (in white) for the diffeomorphisms $\gamma_n$ mapping the Watershed regions $\Omega_n$ (with contour in red) to the disk $\Lambda = \mathrm{B}_2(0,1/2) $ and (bottom) resulting empirical Gabor wavelet transform spectra, for the Barbara image and $n \in \{0,\ldots,14\}$. The Fourier spectrum energies $\mathrm{E}_n$ over the regions $\Omega_n \cup \Omega_{-n}$ are indicated above the wavelet spectra. }
  \label{fig:barbara_watershed_gabor_transform}
\end{figure}

\begin{figure}[htbp]
\centerline{
  \includegraphics[width=\textwidth]{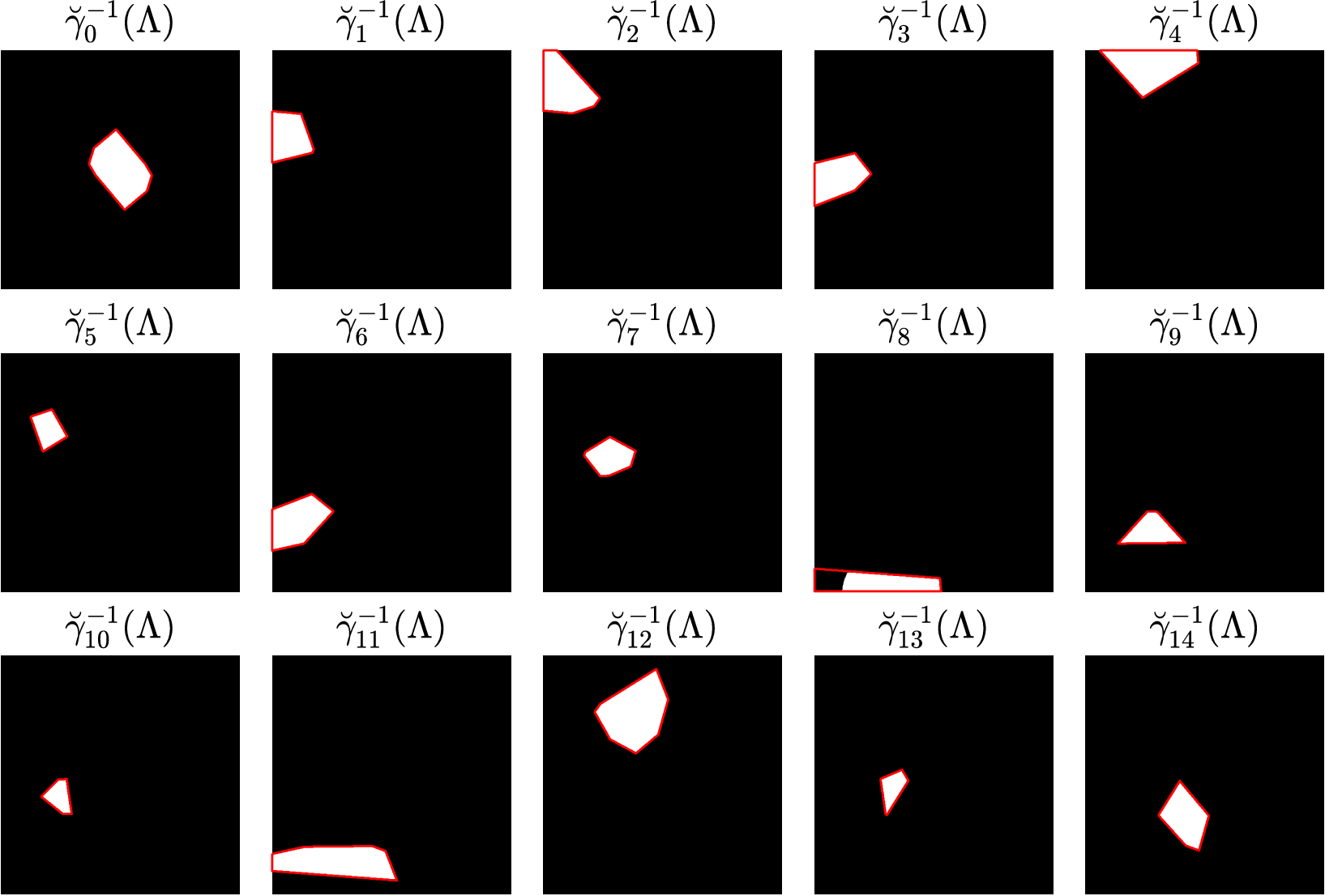} }
  \noindent\rule[5pt]{\textwidth}{0.4pt}
  \centerline{ \includegraphics[width=\textwidth]{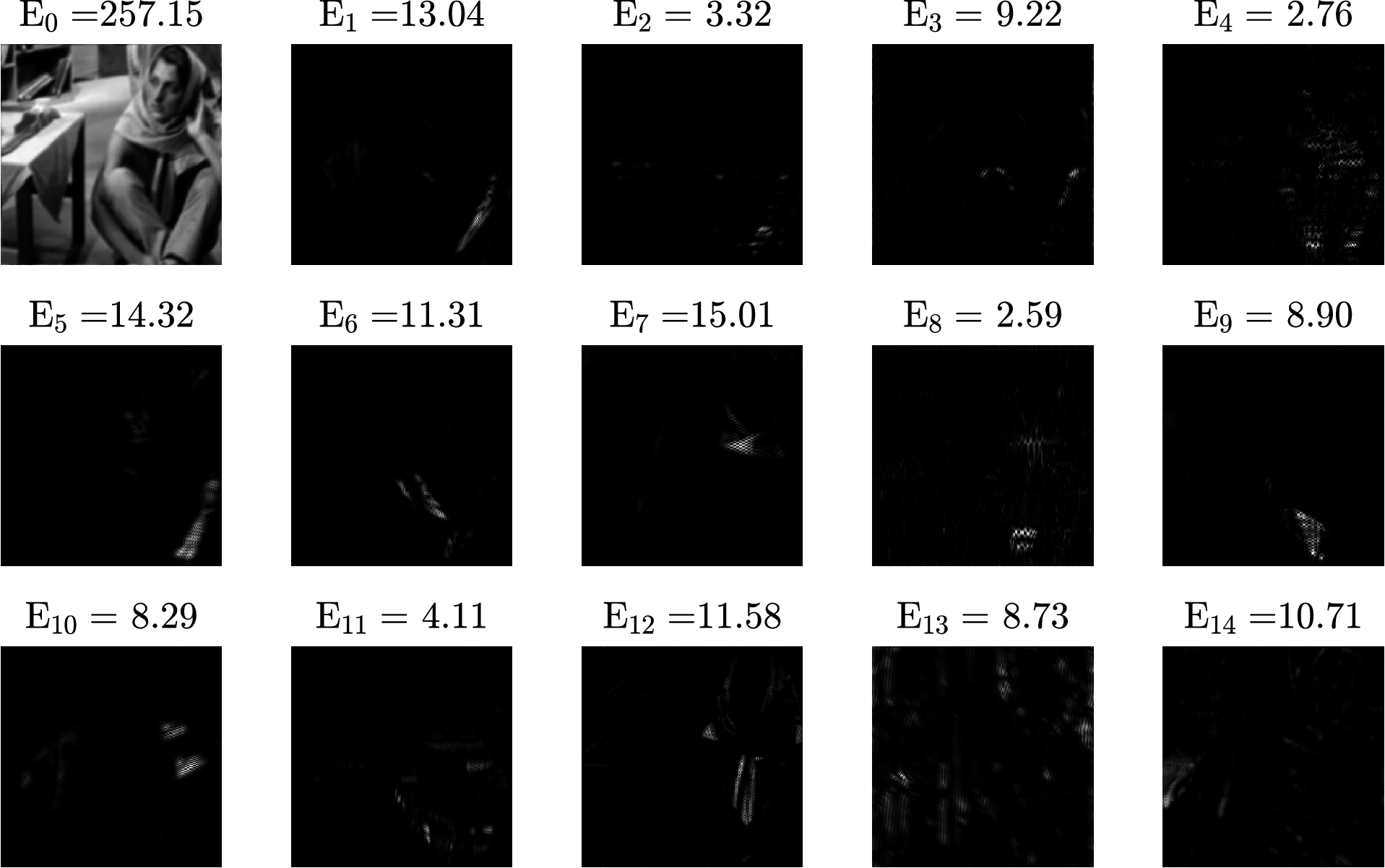} }
 \caption{ {\bf Gabor Voronoi transform of the Barbara image.} (Top) Sets $\breve \gamma_n^{-1}(\Lambda)$ (in white) for the diffeomorphisms $\gamma_n$ mapping the Voronoi regions $\Omega_n$ (with contour in red) to the disk $\Lambda = \mathrm{B}_2(0,1/2) $ and (bottom) resulting empirical Gabor wavelet transform spectra, for the Barbara image and $n \in \{0,\ldots,14\}$. The Fourier spectrum energies $\mathrm{E}_n$ over the regions $\Omega_n \cup \Omega_{-n}$ are indicated above the wavelet spectra. }
  \label{fig:barbara_voronoi_gabor_transform}
\end{figure}

\subsection{2D empirical Shannon wavelets}

The 1D Shannon scaling function is a $sinc$ function given in the Fourier domain by, for every $u \in \R$,
\begin{equation*}
\widehat{\psi}^\mathrm{1D-S}(u)= e^{-\imath\pi(u+\frac{3}{2})}\mathbb{1}_{\left(-\frac{1}{2},\frac{1}{2}\right)}(u).
\end{equation*}
This definition can be extended to the 2D domain as a separable function given by, for every $u=(u_1,u_2) \in \R^2$,
\begin{equation*}
 \widehat{\psi}^\mathrm{S}(u)=\widehat{\psi}^\mathrm{1D-S}(u_1) \, \widehat{\psi}^\mathrm{1D-S}(u_2),
\end{equation*}
which is supported by the square centered in $0$ of side length $1$ denoted $\Lambda = \left(-\frac{1}{2},\frac{1}{2} \right) \times \left(-\frac{1}{2},\frac{1}{2} \right)$.

The following proposition gives guarantees of reconstruction of Shannon empirical wavelet systems.

\begin{proposition}[Shannon wavelet reconstruction]
Assume that the boundaries $\partial \Omega_n$ have measures zero.
Then, the continuous reconstruction is guaranteed for the empirical wavelet systems $\{\psi_n\}_{n\in\Upsilon}$ and $\{\chi_n\}_{n\in\Upsilon}$ induced by $ \widehat{\psi}^\mathrm{S} \circ \gamma_n$, and is given by \Cref{eq:reconstruction,eq:symreconstruction}.
\end{proposition}
\begin{proof}
Let fix $ m\in\Upsilon$.
For every $ \xi \in \Omega_m$, 
\begin{equation*}
\sum_{n\in\Upsilon}\left|\widehat{\psi}_n(\xi)\right|^2 = \sum_{n\in\Upsilon}\vert \mathrm{det}\; J_{\gamma_n}(\xi) \vert \left|\widehat{\psi}(\gamma_n(\xi))\right|^2 = \vert \mathrm{det}\; J_{\gamma_m}(\xi) \vert.
%= \frac{\vert \mathrm{det}\; J_{\gamma_m}(\xi) \vert}{4\pi^2}.
\end{equation*}
Since $\gamma_n$ is a diffeomorphism, it follows that, for $ a.e. \, \xi \in \R^2$, $\vert \mathrm{det}\; J_{\gamma_m}(\xi) \vert > 0$ and therefore
\begin{equation*}
0< \sum_{n\in\Upsilon}\left|\widehat{\psi}_n(\xi)\right|^2 < \infty.
\end{equation*}
This corresponds to the condition of \Cref{theo:icewt}, which gives the reconstruction from a dual frame of $\{\psi_n\}_{n\in\Upsilon}$ using \Cref{eq:reconstruction}.

In addition, for every $\xi\in\R^N$, %$\xi \in \Omega_m$,
\begin{equation*}
\sum_{n\in\Upsilon\setminus\{0\}}\left|\widehat{\psi}_n(\xi)\right| \left|\widehat{\psi}_{-n}(\xi)\right| =0.
\end{equation*}
Then, \Cref{theo:symicewt} gives a dual frame of $\{\chi_n\}_{n\in\Upsilon^+}$ permitting the reconstruction following \Cref{eq:symreconstruction}.
\end{proof}

For the toy image, \Cref{fig:toy_image_watershed_shannon_transform,fig:toy_image_voronoi_shannon_transform} compare the preimage of $\Lambda = \left(-\frac{1}{2},\frac{1}{2} \right) \times \left(-\frac{1}{2},\frac{1}{2} \right)$ under the diffeomorphism estimate $\breve \gamma_n$ to the targeted $\Omega_n$ and show the symmetric empirical Shannon wavelet coefficients induced by $\widehat{\psi}^\mathrm{S} \circ \breve \gamma_n$, for the Watershed and Voronoi partitions, respectively, and $n\in\{0,\ldots,9\}$.
Most of the diffeomorphisms $\gamma_n$ are well estimated for both partitions except for the region $\Omega_4$ of the Watershed partition. Thus, the MSE of the reconstruction from the Watershed and Voronoi transforms are, respectively, $6.24 \times 10^{-10}$ and $3.88 \times 10^{-32}$, 
confirming the higher accuracy of the diffeormophism estimation for the Voronoi partition.
However, the components of the toy image are better separated by the wavelet coefficients associated to the Watershed partition.

\begin{figure}[!ht]
\centerline{
  \includegraphics[width=\textwidth]{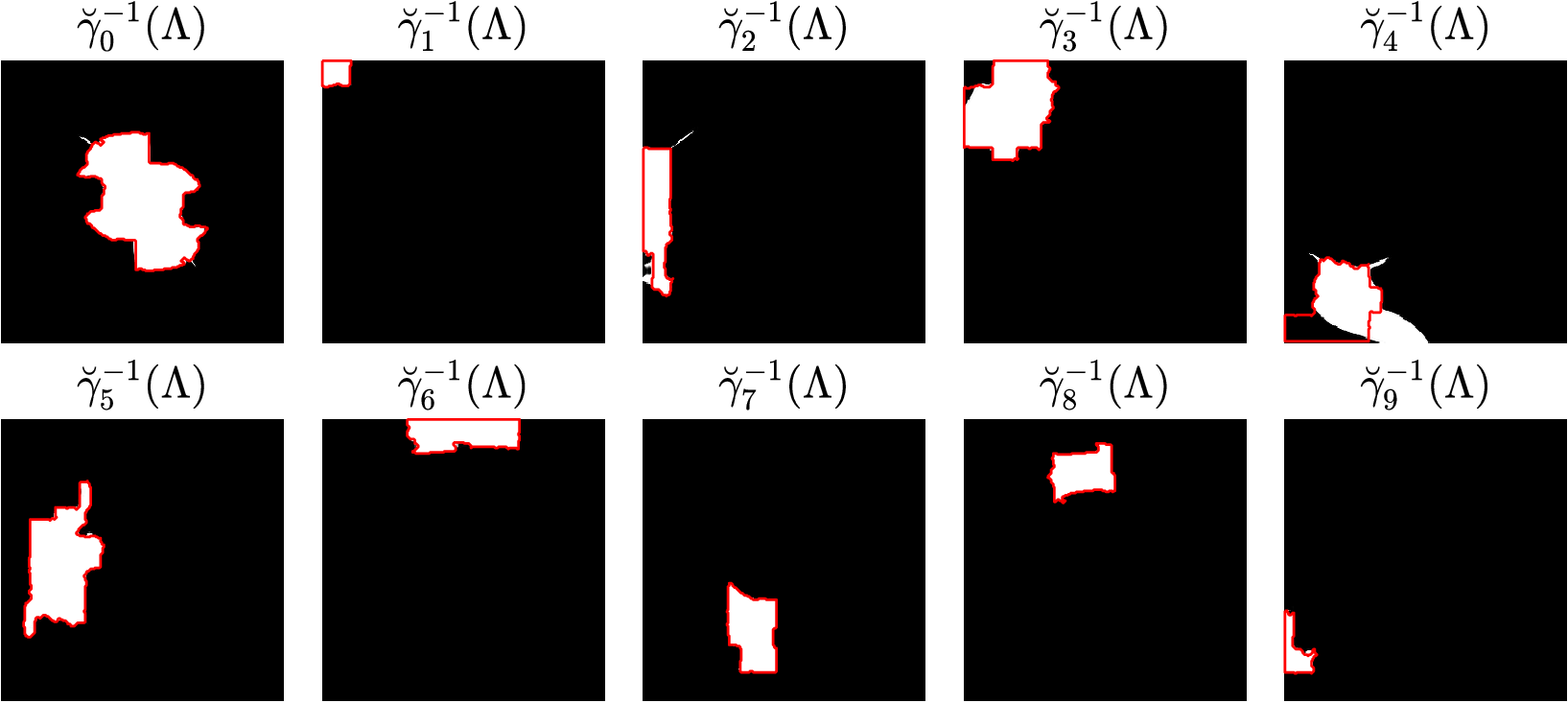} }
  \noindent\rule[5pt]{\textwidth}{0.4pt}
  \centerline{ \includegraphics[width=\textwidth]{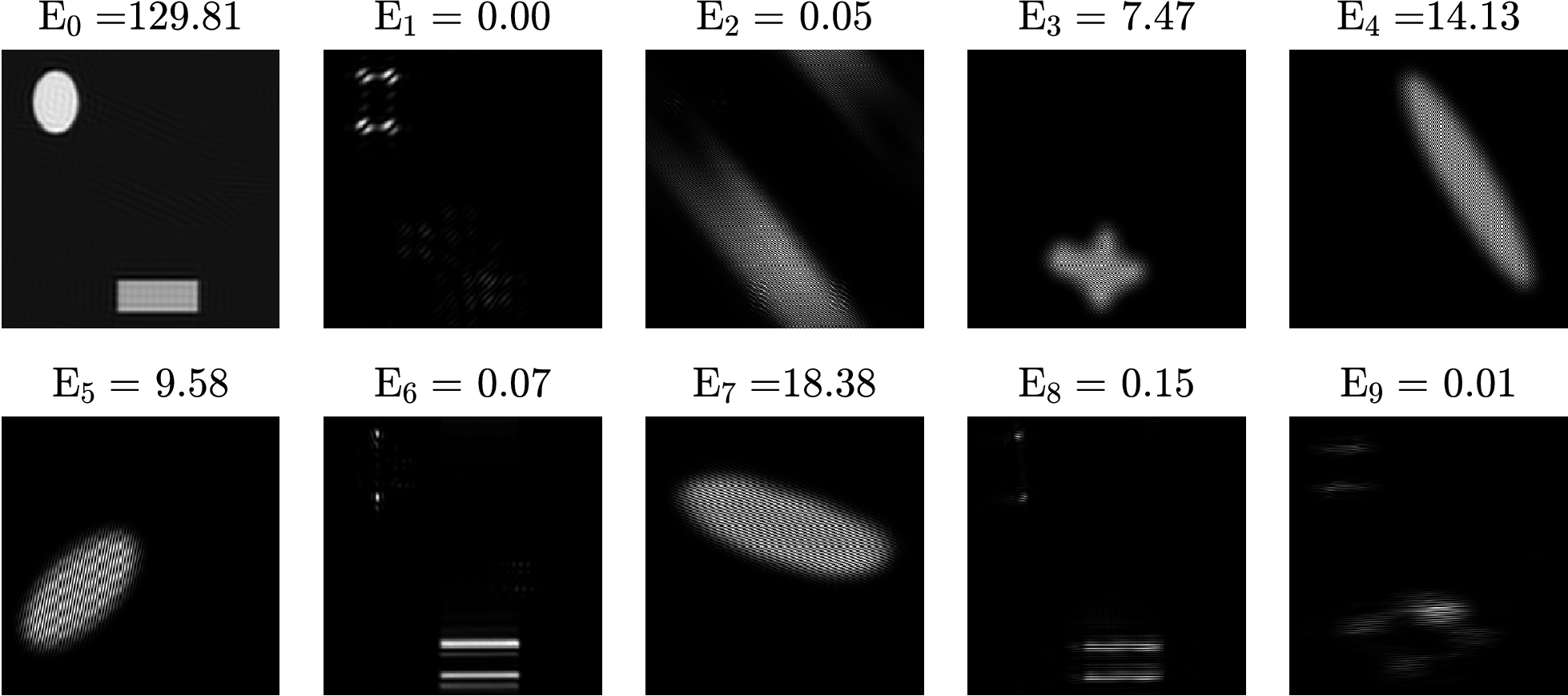} }
 \caption{ {\bf Shannon Watershed transform of the toy image.} (Top) Sets $\breve \gamma_n^{-1}(\Lambda)$ (in white) for the diffeomorphisms $\gamma_n$ mapping the Watershed regions $\Omega_n$ (with contour in red) to the square $\Lambda = \left(-\frac{1}{2},\frac{1}{2} \right) \times \left(-\frac{1}{2},\frac{1}{2} \right) $ and (bottom) resulting empirical Shannon wavelet transform spectra, for the toy image and $n \in \{0,\ldots,9\}$. The Fourier spectrum energies $\mathrm{E}_n$ over the regions $\Omega_n \cup \Omega_{-n}$ are indicated above the wavelet spectra. }
  \label{fig:toy_image_watershed_shannon_transform}
\end{figure}

\begin{figure}[!ht]
\centerline{
  \includegraphics[width=\textwidth]{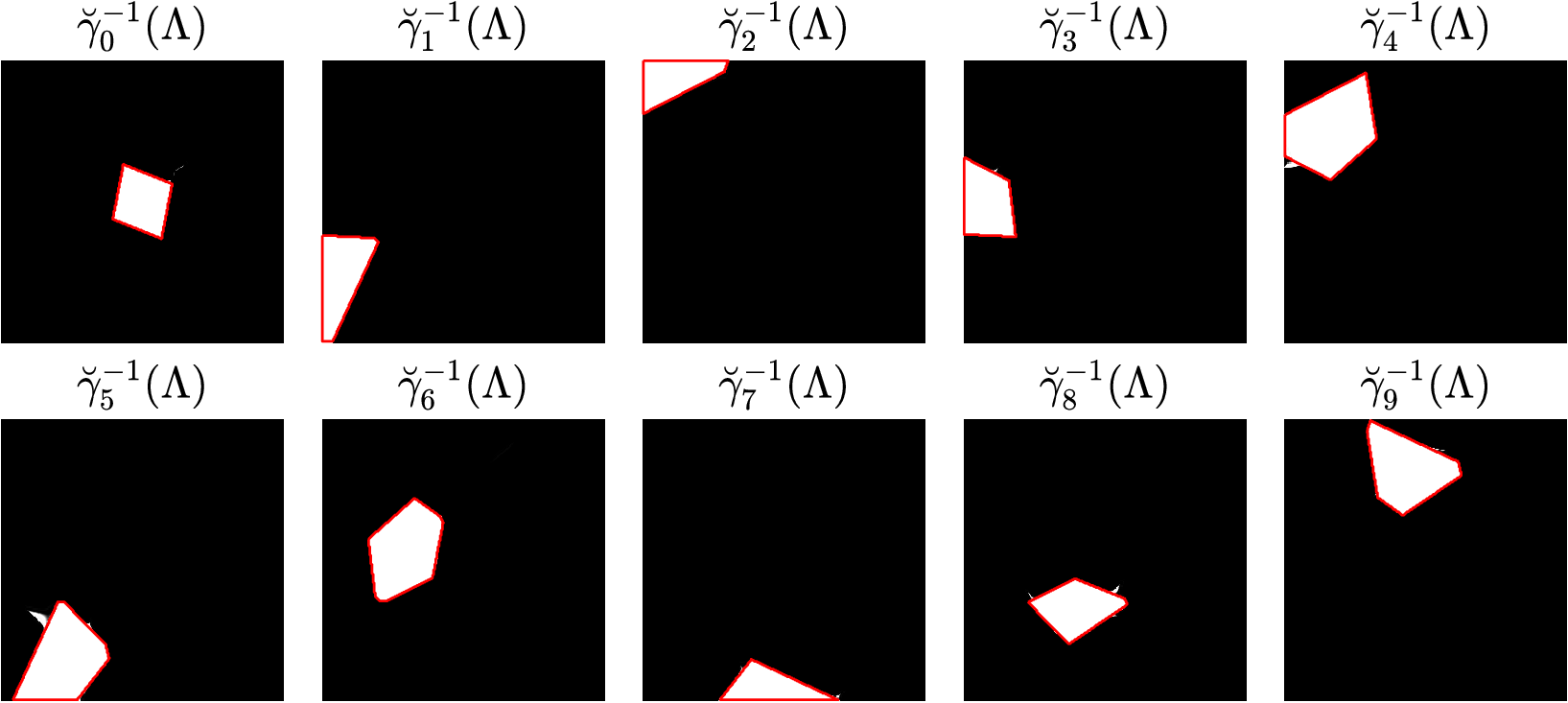} }
  \noindent\rule[5pt]{\textwidth}{0.4pt}
  \centerline{ \includegraphics[width=\textwidth]{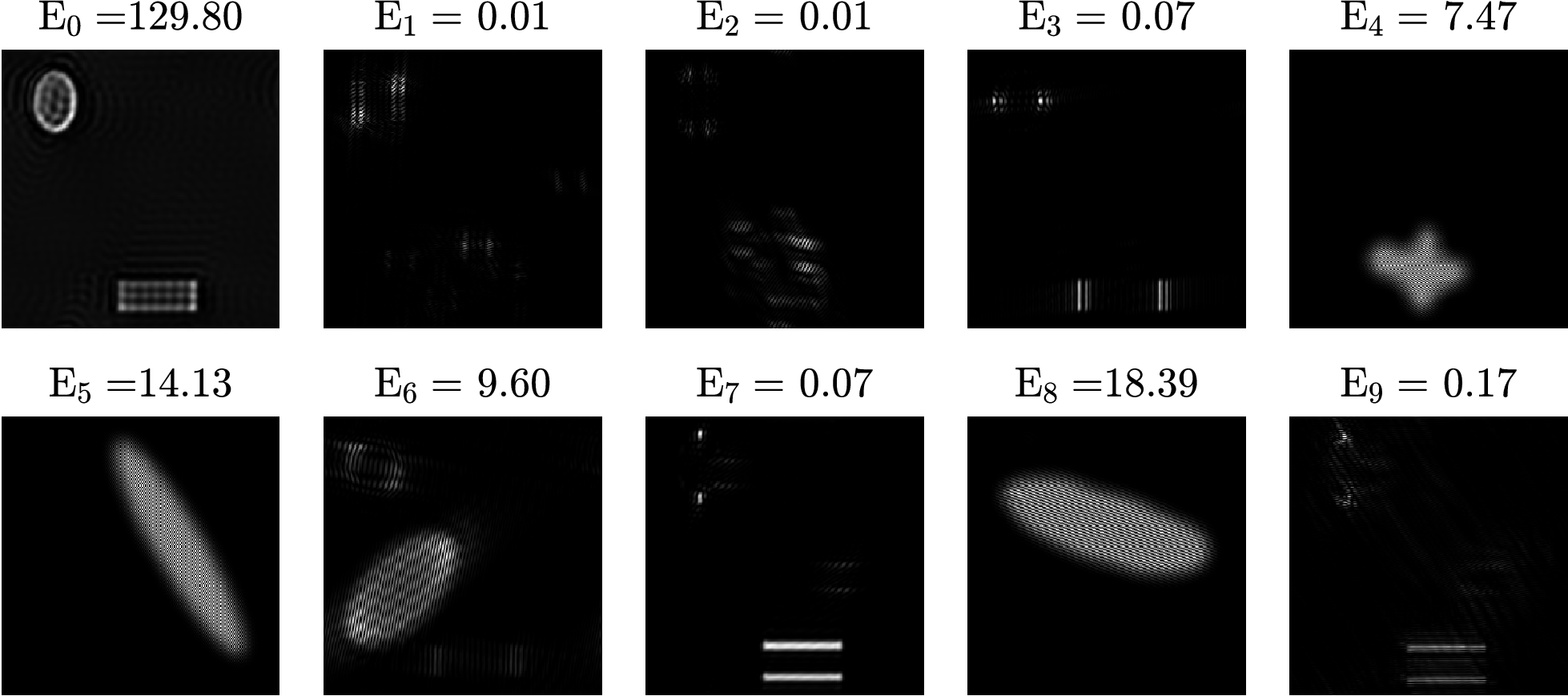} }
 \caption{ {\bf Shannon Voronoi transform of the toy image.} (Top) Sets $\breve \gamma_n^{-1}(\Lambda)$ (in white) for the diffeomorphisms $\gamma_n$ mapping the Voronoi regions $\Omega_n$ (with contour in red) to the square $\Lambda = \left(-\frac{1}{2},\frac{1}{2} \right) \times \left(-\frac{1}{2},\frac{1}{2} \right) $ and (bottom) resulting empirical Shannon wavelet transform spectra, for the toy image and $n \in \{0,\ldots,9\}$. The Fourier spectrum energies $\mathrm{E}_n$ over the regions $\Omega_n \cup \Omega_{-n}$ are indicated above the wavelet spectra. }
  \label{fig:toy_image_voronoi_shannon_transform}
\end{figure}

For the Barbara image, \Cref{fig:barbara_watershed_shannon_transform,fig:barbara_voronoi_shannon_transform} compare the preimage of $\Lambda = \left(-\frac{1}{2},\frac{1}{2} \right) \times \left(-\frac{1}{2},\frac{1}{2} \right)$ under $\breve \gamma_n$ to the targeted $\Omega_n$ along with the symmetric empirical Shannon wavelet coefficients induced by $\widehat{\psi}^\mathrm{S} \circ \breve \gamma_n$, for the Watershed and Voronoi partitions, respectively, and $n\in\{0,\ldots,9\}$. 
The diffeomorphism estimates associated with the Watershed partition show little accuracy for some regions, in particular regions $\Omega_0$, $\Omega_5$ and $\Omega_{9}$. In contrast, the Voronoi partition allows for better estimates despite some singularity due to the irregularity at the corner of the square $\overline{\Lambda}$. 
The reconstruction from the Watershed and Voronoi partitions lead to MSE of, respectively, $9.45 \times 10^{-5}$ and $2.19 \times 10^{-29}$, %$1.89 \times 10^{-5}$ and $7.28 \times 10^{-18}$, 
which quantifies the higher accuracy of diffeomorphism estimation for the Voronoi partition.

\begin{figure}[htbp]
\centerline{
  \includegraphics[width=\textwidth]{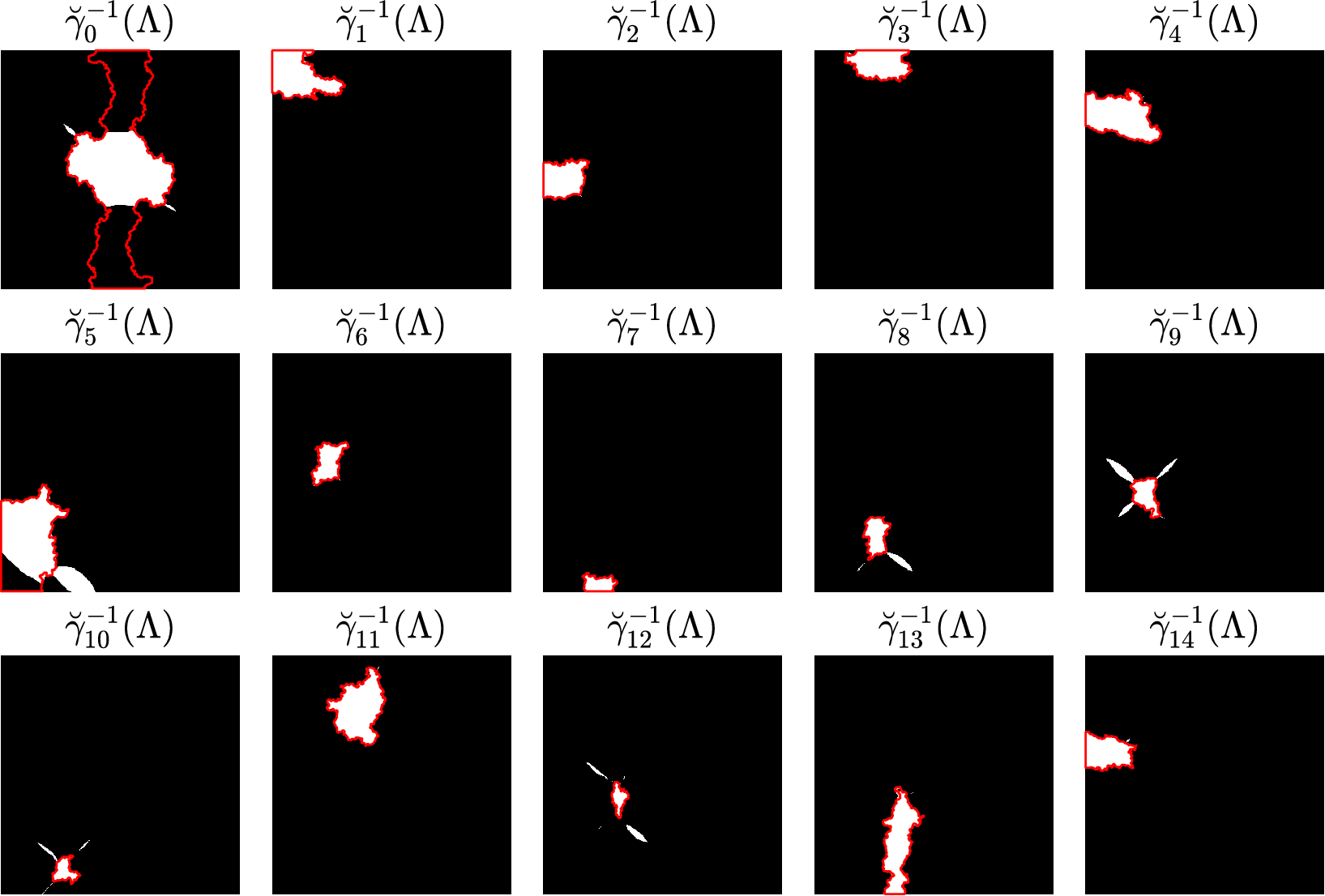} }
  \noindent\rule[5pt]{\textwidth}{0.4pt}
  \centerline{ \includegraphics[width=\textwidth]{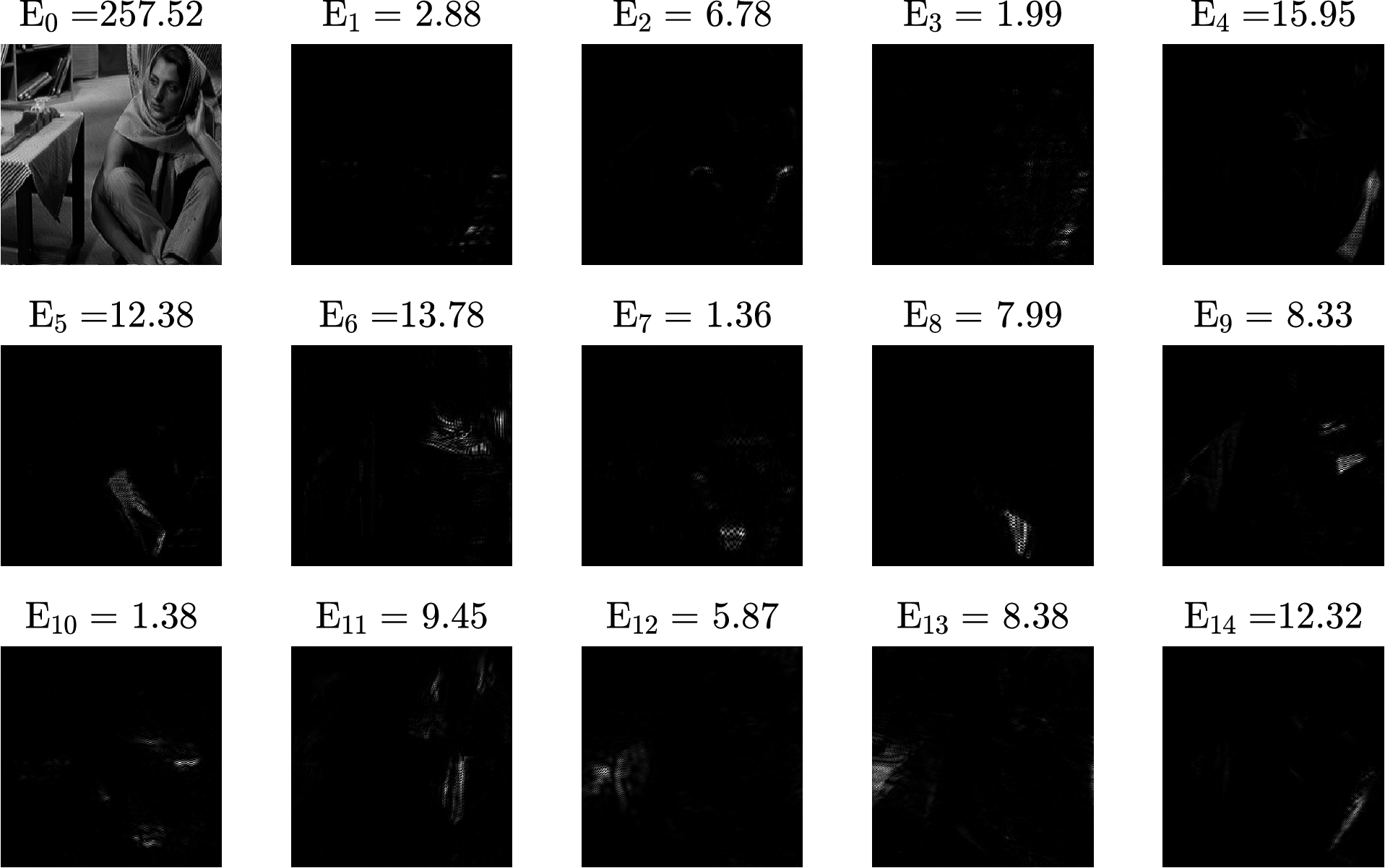} }
 \caption{ {\bf Shannon Watershed transform of the Barbara image.} (Top) Sets $\breve \gamma_n^{-1}(\Lambda)$ (in white) for the diffeomorphisms $\gamma_n$ mapping the Watershed regions $\Omega_n$ (with contour in red) to the square $\Lambda = \left(-\frac{1}{2},\frac{1}{2} \right) \times \left(-\frac{1}{2},\frac{1}{2} \right) $ and (bottom) resulting empirical Shannon wavelet transform spectra, for the Barbara image and $n \in \{0,\ldots,14\}$. The Fourier spectrum energies $\mathrm{E}_n$ over the regions $\Omega_n \cup \Omega_{-n}$ are indicated above the wavelet spectra. }
  \label{fig:barbara_watershed_shannon_transform}
\end{figure}

\begin{figure}[htbp]
\centerline{
  \includegraphics[width=\textwidth]{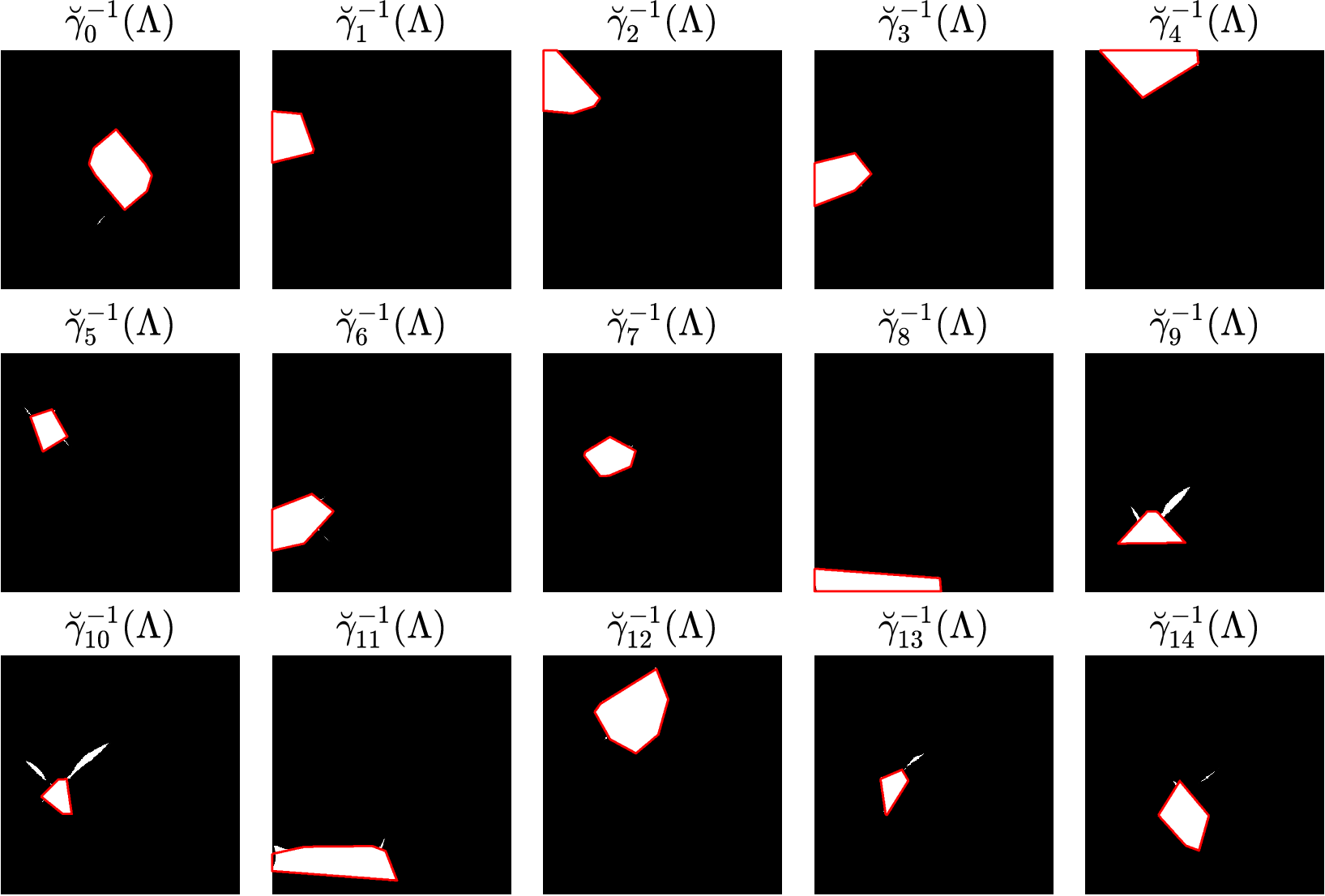} }
  \noindent\rule[5pt]{\textwidth}{0.4pt}
  \centerline{ \includegraphics[width=\textwidth]{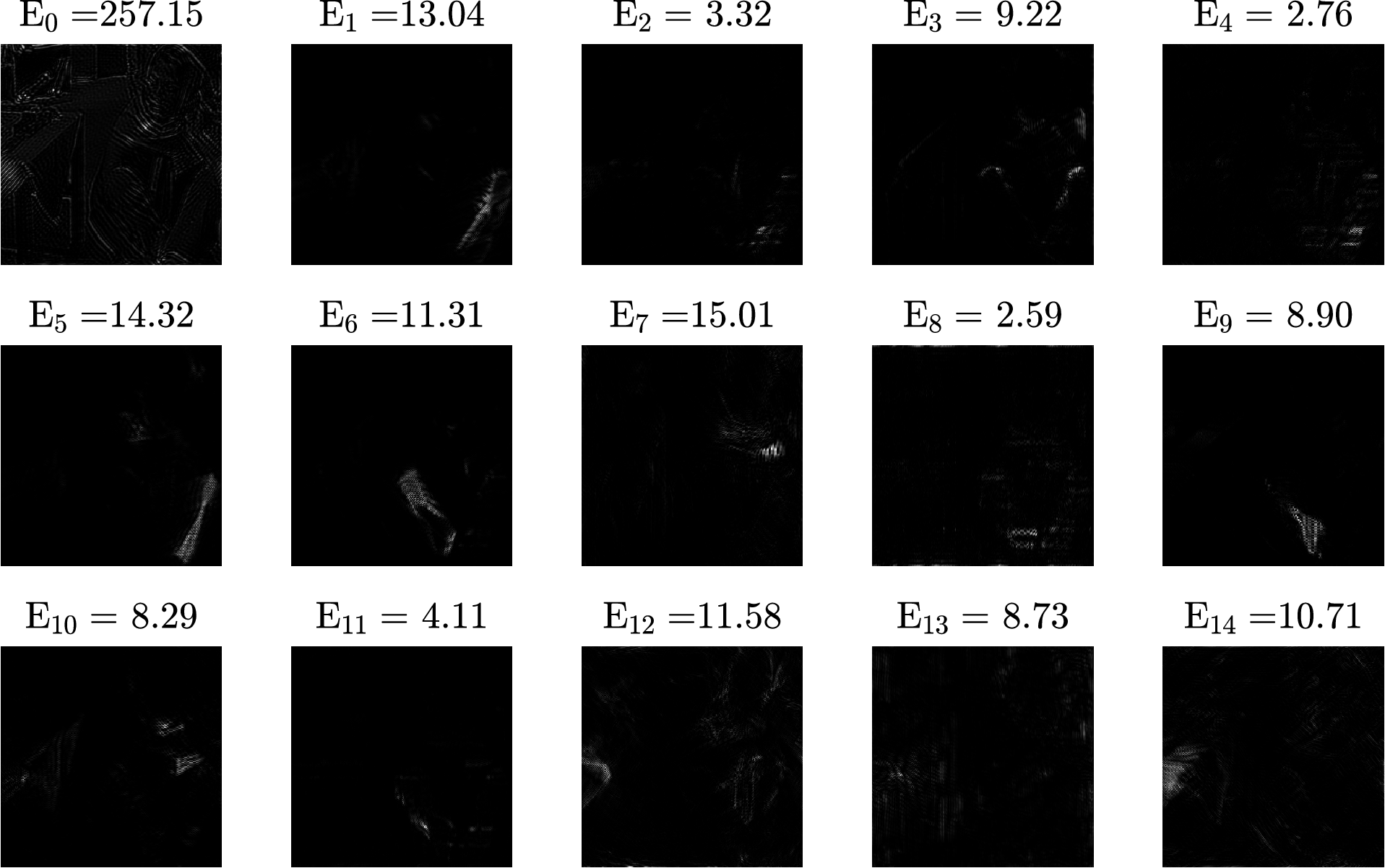} }
 \caption{ {\bf Shannon Voronoi transform of the Barbara image.} (Top) Sets $\breve \gamma_n^{-1}(\Lambda)$ (in white) for the diffeomorphisms $\gamma_n$ mapping the Voronoi regions $\Omega_n$ (with contour in red) to the square $\Lambda = \left(-\frac{1}{2},\frac{1}{2} \right) \times \left(-\frac{1}{2},\frac{1}{2} \right) $ and (bottom) resulting empirical Shannon wavelet transform spectra, for the Barbara image and $n \in \{0,\ldots,14\}$. The Fourier spectrum energies $\mathrm{E}_n$ over the regions $\Omega_n \cup \Omega_{-n}$ are indicated above the wavelet spectra. }
  \label{fig:barbara_voronoi_shannon_transform}
\end{figure}

\subsection{Discussion}

As illustrated in this section, continuous wavelet frames can easily be built theoretically but the resulting numerical transform significantly relies on the estimation of the mappings $\gamma_n$.
Overall, the Voronoi partitions provides regions that are easier to map than the Watershed partitions, particularly when the wavelet kernel's Fourier support is a square. However, the Voronoi partition can lead to a less adapted separation of the harmonic modes, implying that different wavelet coefficients can contain information of the same frequency band.
These observations highlight the need for an efficient diffeomorphism estimator robust to sets $\Lambda$ and $\Omega_n$ for extensions to higher dimensions and applications.
Moreover, a perfect detection and delimitation of harmonic modes is difficult in practice, which suggests to explore several wavelet kernels in applications.

\section{Conclusions}
\label{sec:conclusions}
In this work, we proposed a general formalism to build multidimensional empirical wavelet systems for a large variety of (potentially symmetric) Fourier domain partitions based on continuous wavelet kernels.
Moreover, we showed conditions for the existence of continuous and discrete empirical wavelet frames, which in particular allow to guarantee the reconstruction from the wavelet transforms. Specific wavelet systems based on classic wavelet kernels have also been developed and have been shown to be frames under mild assumptions on the Fourier supports. In addition, the implementation toolbox of these wavelet systems will be made freely available at the time of publication. Future work will focus on estimating robustly and efficiently the diffeomorphisms involved in the proposed definitions in 2D and 3D. Applications will also be considered.

\section*{Acknowledgements}
This work has been sponsored by the Air Force Office of Scientific Research, grant FA9550-21-1-0275.

\bibliographystyle{siamplain}
\bibliography{mybibfile}

\end{document}

%% file: SIAM_NDEWT_shared.tex
\usepackage{lineno,hyperref}
\modulolinenumbers[5]
\usepackage{amsthm}
\usepackage{amsmath}
\usepackage{amssymb}
\usepackage{stmaryrd}
\usepackage{xcolor}
\usepackage{bbold}
\usepackage{subfig}
\usepackage{enumitem}
\usepackage{cleveref}
\usepackage{hyperref}

\usepackage{xr}
\usepackage{ifthen}
\usepackage{fancyhdr}
\usepackage{amsfonts}
\usepackage{graphicx}
\usepackage{epstopdf}
\usepackage{algorithmic}
\ifpdf
  \DeclareGraphicsExtensions{.eps,.pdf,.png,.jpg}
\else
  \DeclareGraphicsExtensions{.eps}
\fi

\newcommand{\email}[1]{\href{mailto:#1}{#1}}

\newtheorem{remark}{Remark}
\newtheorem{hypothesis}{Hypothesis}
\newtheorem{definition}{Definition}
\newtheorem{proposition}{Proposition}
\newtheorem{corollary}{Corollary}
\newtheorem{theorem}{Theorem}
\newtheorem{lemma}{Lemma}
\crefname{hypothesis}{Hypothesis}{Hypotheses}
\newtheorem{example}{Example}[section]

\title{Multidimensional empirical wavelet transform
}

\author{Charles-G\'erard Lucas\thanks{Department of Mathematics \& Statistics, San Diego State University,
5500 Campanile Dr., San Diego, 92182, CA, USA.
  (\email{clucas2@sdsu.edu}, \email{jgilles@sdsu.edu}).
  }
\and J\'er\^ome Gilles\footnotemark[1]
}

\date{}

\usepackage{amsopn}